\documentclass[oneside, a4paper,reqno,11pt]{amsart}

\usepackage{amsmath}
\usepackage{amssymb}
\usepackage{amsfonts}
\usepackage[foot]{amsaddr}
\usepackage{mathtools}
\usepackage{comment}

\usepackage[utf8]{inputenc}
\usepackage[T1]{fontenc}

\usepackage{dsfont}
\usepackage[%
cal=cm,
]
{mathalfa}

\usepackage[dvipsnames,svgnames]{xcolor}
\colorlet{MyBlue}{DodgerBlue!60!Black}
\colorlet{MyGreen}{DarkGreen!85!Black}

\usepackage{fullpage}

\usepackage[font=small,labelfont=bf]{caption}
\usepackage{subfigure}
\usepackage{tikz}
\usetikzlibrary{calc,patterns}

\usepackage{acronym}
\usepackage{booktabs}       % professional-quality tables
\usepackage[shortlabels]{enumitem}
\usepackage{latexsym}
\usepackage{wasysym}
\usepackage{xspace}
\usepackage{longtable}

\usepackage[authoryear,compress,comma]{natbib}
\usepackage{hyperref}
\hypersetup{
colorlinks=true,
linktocpage=true,
pdfstartview=FitH,
breaklinks=true,
pdfpagemode=UseNone,
pageanchor=true,
pdfpagemode=UseOutlines,
plainpages=false,
bookmarksnumbered,
bookmarksopen=false,
bookmarksopenlevel=1,
hypertexnames=true,
pdfhighlight=/O,
urlcolor=MyBlue!60!black,linkcolor=MyBlue!70!black,citecolor=DarkGreen!70!black, % <--- for screen
pdftitle={},
pdfauthor={},
pdfsubject={},
pdfkeywords={},
pdfcreator={pdfLaTeX},
pdfproducer={LaTeX with hyperref}
}

\numberwithin{equation}{section}  %numberwithin goes before cleverefs when using hyperref
\usepackage[sort&compress,capitalize,nameinlink]{cleveref}
\crefname{app}{Appendix}{Appendices}

\crefrangeformat{equation}{\upshape(#3#1#4)\textendash(#5#2#6)}

\newcommand{\R}{\mathbb{R}}

\newcommand{\N}{\mathbb{N}}
\newcommand{\reals}{\mathbb{R}}

\DeclareMathOperator{\Card}{\debug{card}}

\DeclareMathOperator{\diag}{\debug{diag}}

\DeclarePairedDelimiter{\braces}{\{}{\}}
\DeclarePairedDelimiter{\bracks}{[}{]}
\DeclarePairedDelimiter{\parens}{(}{)}

\DeclarePairedDelimiterX{\braket}[2]{\langle}{\rangle}{#1,#2}

\DeclarePairedDelimiterX{\inner}[2]{\langle}{\rangle}{#1,#2}
\DeclarePairedDelimiterX{\setdef}[2]{\{}{\}}{#1:#2}

\DeclarePairedDelimiterXPP{\probof}[1]{\Prob}{(}{)}{}{%

#1}

\DeclarePairedDelimiterXPP{\exof}[1]{\Expect}{[}{]}{}{%

#1}

\usepackage[textwidth=30mm]{todonotes}

\newcommand{\debug}[1]{#1}

\theoremstyle{plain}
\newtheorem{theorem}{Theorem}
\newtheorem{corollary}[theorem]{Corollary}
\newtheorem*{corollary*}{Corollary}
\newtheorem{lemma}[theorem]{Lemma}
\newtheorem{proposition}[theorem]{Proposition}

\theoremstyle{definition}
\newtheorem{definition}[theorem]{Definition}
\newtheorem*{definition*}{Definition}

\newtheorem*{hypothesis*}{Hypothesis}

\theoremstyle{remark}
\newtheorem{remark}{Remark}
\newtheorem*{remark*}{Remark}
\newtheorem*{notation*}{Notational remark}
\newtheorem{example}{Example}

\numberwithin{theorem}{section}
\numberwithin{remark}{section}
\numberwithin{example}{section}

\DeclareMathOperator{\Prob}{\mathsf{\debug{P}}}
\DeclareMathOperator{\Expect}{\mathsf{\debug{E}}}

\newcommand{\cost}{\debug c}
\newcommand{\costprof}{\boldsymbol{\cost}}
\newcommand{\Cost}{\debug C}

\newcommand{\eqcost}{\debug \lambda}

\newcommand{\alphaint}{\debug \theta}

\newcommand{\eq}[1]{#1^{\ast}}
\newcommand{\opt}[1]{\widetilde#1}
\newcommand{\activ}[1]{\widehat#1}

\DeclareMathOperator{\SC}{\mathsf{\debug{SC}}}
\DeclareMathOperator{\PoA}{\mathsf{\debug{PoA}}}

\newcommand{\graph}{\mathcal{\debug G}}
\newcommand{\vertices}{\mathcal{\debug V}}
\newcommand{\edges}{\mathcal{\debug E}}

\newcommand{\vertex}{\debug v}
\newcommand{\vertexalt}{\debug w}

\newcommand{\edge}{\debug e}

\newcommand{\source}{\debug O}

\newcommand{\sink}{\debug D}

\newcommand{\inneigh}{\mathcal{\debug{N}^{-}}}
\newcommand{\outneigh}{\mathcal{\debug{N}^{+}}}

\newcommand{\rate}{\debug \mu}

\newcommand{\flow}{\debug f}
\newcommand{\flows}{\mathcal{\debug F}}
\newcommand{\flowprof}{\boldsymbol{\flow}}

\newcommand{\load}{\debug x}

\newcommand{\loads}{\mathcal{\debug X}}
\newcommand{\loadprof}{\boldsymbol{\load}}

\newcommand{\nRoutes}{\debug P}
\newcommand{\routes}{\mathcal{\debug \nRoutes}}
\newcommand{\route}{\debug p}
\newcommand{\routealt}{\route'}

\newcommand{\argdot}{\,\cdot\,}
\newcommand{\card}[1]{\Card\parens{#1}}
\newcommand{\diff}{\ \textup{d}}

\newcommand{\zerovec}{\boldsymbol{\debug 0}}

\newcommand{\coeffa}{\debug a}
\newcommand{\coeffb}{\debug b}
\newcommand{\coeffbprof}{\boldsymbol{\coeffb}}
\newcommand{\coeffalpha}{\debug \alpha}
\newcommand{\coeffbeta}{\debug \beta}
\newcommand{\coeffgamma}{\debug \gamma}
\newcommand{\coeffdelta}{\debug \delta}

\newcommand{\coeffzeta}{\debug \zeta}
\newcommand{\coeffeta}{\debug \eta}

\newcommand{\degr}{\debug d}
\newcommand{\degralt}{\debug q}

\newcommand{\lagrang}{\mathcal{\debug L}}

\newcommand{\perturb}{\debug \varphi}

\newcommand{\valueW}{\debug V}
\newcommand{\vinf}{\debug v}

\newcommand{\nedges}{\debug m}
\newcommand{\npaths}{\debug n}

\newcommand{\eqcostedge}{\debug \tau}
\newcommand{\eqcostedgevec}{\boldsymbol{\eqcostedge}}
\newcommand{\eqcostvertex}{\debug T}
\newcommand{\eqcostvertexvec}{\boldsymbol{\eqcostvertex}}

\newcommand{\epincid}{\debug Z}
\newcommand{\diaga}{\debug \Gamma}

\newcommand{\wvec}{\boldsymbol{\debug w}}

\newcommand{\zvec}{\boldsymbol{\debug z}}

\newcommand{\run}{\debug k}
\newcommand{\irun}{\debug i}

\newcommand{\zvar}{\debug z}
\newcommand{\matrA}{\debug A}

\newcommand{\sumbvec}{ \boldsymbol{\debug d}}

\newcommand{\prim}{\mathsf{\debug P}}
\newcommand{\dual}{\mathsf{\debug D}}
\newcommand{\solset}{\mathsf{\debug S}}
\newcommand{\incrload}{\debug u}
\newcommand{\incrloadvec}{\boldsymbol{\incrload}}
\newcommand{\incrcostedge}{\debug s}
\newcommand{\incrcostvertex}{\debug \delta}
\newcommand{\incrcostvertexvec}{\boldsymbol{\incrcostvertex}}

\newcommand{\Tstrut}{\rule{0pt}{4.0ex}}       % Top strut
\newcommand{\Bstrut}{\rule[-2.5ex]{0pt}{0pt}} % Bottom strut
\newcommand{\Bstrutsmall}{\rule[-1.5ex]{0pt}{0pt}} % Bottom strut

\DeclareMathOperator{\BPR}{\mathsf{\debug{BPR}}}

\newacro{ACG}{atomic congestion game}
\newacro{ACGSD}{atomic congestion game with stochastic demand}

\newacro{PoA}{price of anarchy}
\newacro{PoS}{price of stability}
\newacro{SC}{social cost}
\newacro{SEC}{social expected cost}
\newacro{SO}{social optimum}
\newacro{SOC}{socially optimum cost}

\newacro{NE}{Nash equilibrium}
\newacroplural{NE}[NE]{Nash equilibria}
\newacro{BNE}{Bayesian Nash equilibrium}
\newacroplural{BNE}[BNE]{Bayesian Nash equilibria}
\newacro{PNE}{pure Nash equilibrium}
\newacroplural{PNE}[PNE]{pure Nash equilibria}

\newacro{WE}{Wardrop equilibrium}
\newacroplural{WE}[WE]{Wardrop equilibria}

\newacro{KKT}{Karush\textendash Kuhn\textendash Tucker}
\newacro{OD}[OD]{origin-destination}
\newacro{BPR}{Bureau of Public Roads}

\newacro{SP}{series-parallel}

\begin{document}

\title{The Price of Anarchy in Routing Games as a Function of the Demand}

\author{Roberto Cominetti}
\email{roberto.cominetti@uai.cl}
\address{Facultad de Ingeniería y Ciencias\\Universidad Adolfo Ibáñez\\
Santiago\\CHILE}
\author{Valerio Dose}
\email{vdose@luiss.it}
\address{Dipartimento di Economia e Finanza\\Luiss University\\
Viale Romania 32\\00197 Roma\\ITALY}
\author{Marco Scarsini}
\email{marco.scarsini@luiss.it}
\address{Dipartimento di Economia e Finanza\\Luiss University\\
Viale Romania 32\\00197 Roma\\ITALY}


\begin{abstract}
The \acl{PoA} has become a standard measure of the efficiency of equilibria in games.
Most of the literature in this area has focused on establishing worst-case bounds for specific classes of games, 
such as routing games or more general congestion games. Recently, the \acl{PoA} in routing games has been 
studied as a function of the traffic demand, providing asymptotic results in light and heavy traffic. 
The aim of this paper is to study the \acl{PoA} in nonatomic routing games in the intermediate region of the demand. 
To achieve this goal, we begin by establishing some smoothness properties of Wardrop equilibria and social optima for general
smooth costs. In the case of affine costs we show that the equilibrium is piecewise linear, with 
break points at the demand levels at which the set of active paths changes. We prove that the number 
of such break points is finite, although it can be exponential in the size of the network. Exploiting a scaling 
law between the equilibrium and the social optimum, we derive a similar behavior for the optimal flows.
We then prove that in any interval between break points
the \acl{PoA} is smooth and  it is either monotone (decreasing or increasing) over the full interval, or it
decreases up to a certain minimum point in the interior of the interval 
and increases afterwards.
We deduce that for affine costs the maximum of the \acl{PoA} can only occur at the break points. 
For general costs we provide counterexamples showing that the set of break points is not always finite.

\end{abstract}

\maketitle

\tableofcontents

%----------------------------------------------------------------------
%%% INTRODUCTION
%----------------------------------------------------------------------

\pagebreak

\section{Introduction}
\label{se:intro}

Nonatomic routing games provide a model  for the distribution of traffic over networks with a large number of drivers, 
each one representing a negligible fraction of the total demand. The model is based on a directed graph with one or 
more origin-destination pairs, and the costs are identified with the delays incurred to go from origin to destination.
The delay on an edge  is a nondecreasing function of the load of players on that edge, and the delay of a path is additive over its edges.
The standard solution concept for such nonatomic games is the Wardrop equilibrium, according to which the traffic  in each \ac{OD} pair travels along paths of minimum delay. 
The aggregate social cost experienced by the whole traffic is therefore the product of these minimal delays  multiplied by the corresponding  traffic demands, summed over all \ac{OD} pairs.

Equilibria are known to be inefficient, so that a social planner would be able to reduce the social cost by redirecting flows along the network.
The most common measure of inefficiency is the \acfi{PoA}\acused{PoA}, that is, the ratio of the  social cost at equilibrium over the minimum social cost.
For nonatomic congestion games with affine costs, the value of the \ac{PoA} is bounded above by $4/3$, and this bound is known to be sharp  \citep{RouTar:JACM2002}. 
On the other hand, for a large class of cost functions, including all polynomials, the \ac{PoA} converges to $1$ as the traffic demand goes either to $0$ or to infinity.
In other words, equilibria tend to perfect efficiency both in light and heavy traffic \citep{ColComMerSca:OR2020,WuMohCheXu:OR2021}.

Empirical studies have shown that in real networks, and for intermediate levels of the demand, the \ac{PoA} tends to oscillate and 
often does not reach the worst case bounds. \begin{figure} [ht] 
\subfigure[A simple network with affine costs]  
{  
\begin{tikzpicture}[scale=.8]  
   \node[shape=circle,draw=black,line width=1pt,minimum size=0.5cm] (O) at (-4,0)  { $\source$}; 
   \node[shape=circle,draw=black,line width=1pt,minimum size=0.5cm] (v1) at (0,0)  {$\vertex_{1}$}; 

   \node[shape=circle,draw=black,line width=1pt,minimum size=0.5cm] (v2) at (-2,3)  {$\vertex_{2}$}; 
   \node[shape=circle,draw=black,line width=1pt,minimum size=0.5cm] (v3) at (2,3)  {$\vertex_{3}$}; 
   \node[shape=circle,draw=black,line width=1pt,minimum size=0.5cm] (D) at (4,0)  {$\sink$}; 
    
   \draw[line width=1pt,->] (O) to   node[midway, fill=white] {$\load$} (v1);
   \draw[line width=1pt,->] (O) to   node[midway, fill=white] {$1$} (v2);
   \draw[line width=1pt,->] (v1) to   node[midway, fill=white] {$6$} (D);
   \draw[line width=1pt,->] (v1) to   node[midway, fill=white] {$0$} (v2);
   \draw[line width=1pt,->] (v1) to   node[midway, fill=white] {$2$} (v3);
   \draw[line width=1pt,->] (v2) to   node[midway, fill=white] {$\load$} (v3);
   \draw[line width=1pt,->] (v3) to   node[midway, fill=white] {$\load$} (D);

\end{tikzpicture}
\label{fi:oscillations-a}
}  
\hspace{1cm}
\subfigure[\ac{PoA} as a function of the traffic demand]  
{  
\includegraphics[width=0.4\textwidth]{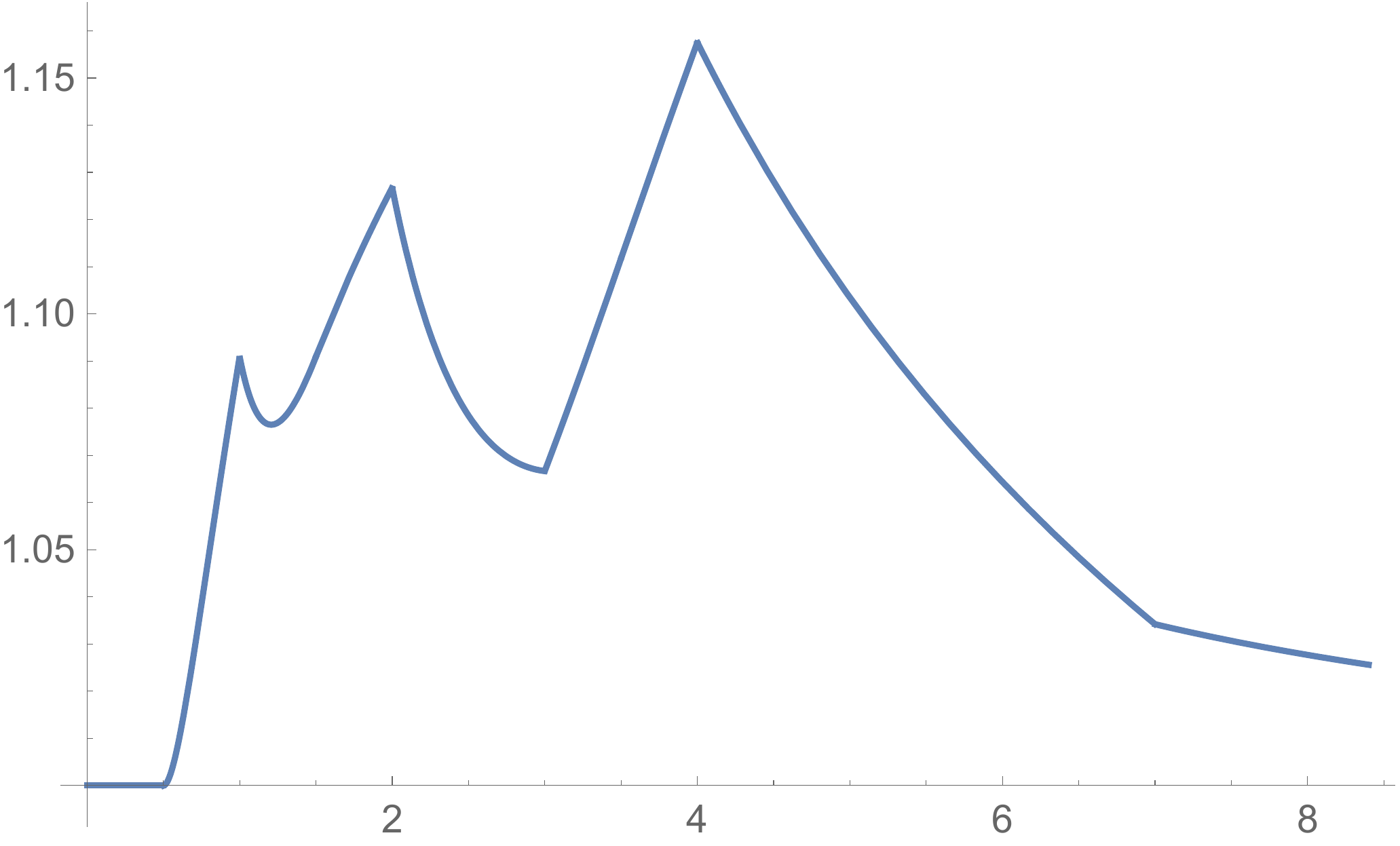}
\label{fi:oscillations-b}
}
\caption{Oscillations of the \acl{PoA}}
\label{fi:oscillations}
\end{figure}
Figure \ref{fi:oscillations} shows a typical profile of the \ac{PoA} as a 
function of the traffic demand. It starts at 1 for low levels of traffic, then it exhibits some oscillations with 
a number of nonsmooth spikes, and eventually it decreases smoothly back to 1 in the highly congested regime.
The shape and number of these oscillations and spikes is the object of this paper.

%----------------------------------------------------------------------
%%% OUR CONTRIBUTION
%----------------------------------------------------------------------

\subsection{Our contribution}

We consider nonatomic routing games over a network with a single \ac{OD},
and we study the behavior of the \acl{PoA} as a function of the traffic demand.

To achieve our goal, we need some general results on the continuity and monotonicity of the equilibrium costs and flows. 
We resort to the classical result in \citet{BecMcGWin:Yale1956} according to which a Wardrop equilibrium 
is a solution of a convex optimization program. We show that the optimal value of this program is convex 
and continuously differentiable as a function of the demand, and its derivative is precisely the equilibrium cost.  
A similar result is established for the minimum social cost.

When the costs have a strictly positive derivative, we show that the equilibrium loads and the \ac{PoA} are  in fact 
$C^1$ at each demand level that is \emph{regular}, in the sense that all the optimal paths carry a strictly positive flow. 
This regularity fails in particular at the so-called $\activ{\edges}$-breakpoints, which are demand levels at which the set of  
shortest paths at equilibrium changes. 

For affine costs we bypass regularity and we show directly that the equilibrium cost is piecewise linear and differentiable except at $\activ{\edges}$-breakpoints. 
The crucial property is that if the set of shortest paths is the same at two demand levels, then this set remains optimal in between. 
From here it follows that the number of $\activ{\edges}$-breakpoints is finite, though it can be exponentially large.  
It also follows that between $\activ{\edges}$-breakpoints the \ac{PoA} is differentiable and either is monotone, or is first decreasing and then increasing, so it has a unique minimum in the interior of the interval.
From this we conclude that the maximum of the \ac{PoA} is attained at an $\activ{\edges}$-breakpoint.

We finally present several examples showing how these properties might fail for general costs.

%----------------------------------------------------------------------
%%% RELATED WORK
%----------------------------------------------------------------------

\subsection{Related Work} 

The standard solution concept in nonatomic routing games is due to \citet{War:PICE1952}.
Its mathematical properties were first studied  by \citet{BecMcGWin:Yale1956};  early algorithms for computing equilibria were proposed by \citet{Tom:OR1966} for affine costs, and by \citet{DafSpa:JRNBSB1969} for general convex 
costs. For historical surveys on the topic we refer to \citet{FloHea:HORMS1995} as well as \citet{CorSti:EORMS2011}.

Several papers have considered the sensitivity of the equilibrium flows and costs with respect to variations of the traffic demand.
\citet{Hal:TS1978} showed that an increase in the demand of one \ac{OD} pair always increases the equilibrium cost for that pair, 
whereas \citet{Fis:TRB1979} observed that the cost on a different \ac{OD} could be reduced.
These questions were further explored by \citet{DafNag:MP1984}. 
In a different direction, \citet{Pat:TS2004} characterized the existence of directional derivatives for the equilibrium.
\citet{JosPat:TRB2007} proved that equilibrium costs are always directionally differentiable, whereas 
equilibrium edge loads not always are.
\citet{EngFraOlb:TCS2010} showed that there exist single \ac{OD} network games where an $\varepsilon$-increase in the traffic demand produces a global migration of traffic from one set of equilibrium paths to a disjoint set of paths, but, nevertheless, the load on each edge changes at most by $\varepsilon$.
Moreover, if the cost functions are polynomials of degree at most $\degr$, then the equilibrium costs increase at most of a multiplicative factor $(1+\varepsilon)^{\degr}$.
\citet{TakKwo:OL2020} extended this last result to games with multiple \ac{OD} pairs.

The recognition that selfish behavior produces social inefficiency goes back at least to \citet{Pig:Macmillan1920}. 
A measure to quantify this inefficiency was proposed by \citet{KouPap:STACS1999}, 
by considering  the ratio of the social cost of the worst equilibrium over the optimum social cost.
It was termed \acl{PoA} by \citet{Pap:ACMSTC2001}.
A \ac{PoA} close to one indicates efficiency of the equilibria of the game, whereas a high \ac{PoA} implies that, in the worst scenario, strategic behavior can lead to significant social inefficiency. 
Most of the subsequent literature established sharp bounds for the \ac{PoA} for specific classes of games,
notably for congestion games and in particular for routing games.
\citet{RouTar:JACM2002} showed that in every nonatomic congestion game with affine costs the \ac{PoA} is bounded above by $4/3$. 
Moreover, this bound is sharp and attained in a traffic game with a simple two-edge parallel network.
\citet{Rou:JCSS2003} generalized this result to polynomial functions of maximum degree $\degr$  showing that the \ac{PoA} grows as $\Theta(\degr/\log\degr)$.
\citet{DumGai:INE2006} refined the result when the cost functions are sums of monomials whose degrees are between $\degralt$ and $\degr$.
\citet{RouTar:GEB2004} extended the analysis to all differentiable cost function $\cost$
such that $\load\cost(\load)$ is convex.
Less regular costs and different optimizing criteria for the social cost were studied by \citet{CorSchSti:GEB2008,CorSchSti:MOR2004,CorSchSti:OR2007}.

Some papers took a more applied view and studied the actual value of the \ac{PoA} in real networks. 
\citet{YouGasJeo:PRL2008,YouGasJeo:PRL2009} dealt with traffic in Boston, London, and New York, 
and noted that the \ac{PoA} exhibits a similar pattern in the three cities: it is $1$ when traffic is light, it oscillates in the central region and then goes back to $1$ when traffic increases.
A similar behavior was observed by \citet{OHaConWat:TRB2016}, who---taking an approach that relates to the one in this paper---showed how an expansion and/or retraction of the routes used at equilibrium affects the behavior of the \ac{PoA}. 
Experimentally, \citet{MonBenPil:WINE2018} studied the commuting behavior of a large number of Singaporean students 
and concluded that the \ac{PoA} is overall low and far from the worst case scenarios. 

An analytical justification for the asymptotic efficiency of the \ac{PoA} in light and heavy traffic was presented in \citet{ColComSca:TOCS2019,ColComMerSca:OR2020,WuMohCheXu:OR2021}.  
\citet{ColComSca:TOCS2019} considered the case of single \ac{OD} parallel networks and proved that, in heavy traffic, the \ac{PoA} converges to one when the cost functions are regularly varying.
Their results were extended in various directions in \citet{ColComMerSca:OR2020}, considering general networks and analyzing both the light and heavy traffic asymptotics. 
A different technique, called \emph{scalability}, was used by \citet{WuMohCheXu:OR2021} to study the case of heavy traffic. 

\citet{WuMoh:arXiv2020} considered issues that are quite close to the one examined here. 
They defined a metric on the space of nonatomic congestion games that share the same network, commodities, and strategies, but differ in terms of demands and cost functions. 
Using this metric they showed that the \ac{PoA} is a continuous function of both the demand and the costs.
Then they performed a sensitivity analysis of the \ac{PoA} with respect to variations of the game in terms of this metric.

In a very interesting recent paper \citet{KliWar:MORfrth} \citep[see also the conference version][]{KliWar:SODA2019}
consider nonatomic routing games with piecewise linear costs and present algorithms to track the full path 
of \aclp{WE}  when the demands vary proportionally along a fixed direction. 
These algorithms are based on (positive or negative) electrical flows on undirected graphs and are then suitably adapted to positive flows on directed graphs.
The connection between their paper and  ours will be discussed in \cref{se:affine}.

The behavior of the \ac{PoA} as a function of a different parameter was studied by \citet{ComScaSchSti:EC2019}.
In that case the parameter of interest was the probability that players actually take part in the game. 
\citet{ColKliSca:ICALP2018} studied the possibility of achieving efficiency in routing games via the use of tolls, when the demand can vary,
while \citet{GemKouMonPapPil:ISTACS2019} analyzed the income inequality effects of reducing the \ac{PoA} via tolls.

%----------------------------------------------------------------------
%%% ORGANIZATION
%----------------------------------------------------------------------

\subsection{Organization of the paper}

In \cref{se:model} we recall the model of nonatomic routing games, and we set the notations
and the standing assumptions. \cref{sec:prelim} investigates the smoothness of the equilibrium costs and of the 
\ac{PoA} as a function of the demand, for general nondecreasing smooth costs. 
The behavior of the \ac{PoA} for  affine costs is studied  in \cref{se:affine}. 
\cref{se:examples} presents various examples.

%----------------------------------------------------------------------
%%% MODEL
%----------------------------------------------------------------------

\section{The nonatomic congestion model}
\label{se:model}

We consider a  nonatomic routing game with a single origin-destination pair. 
The network is described by a directed multigraph $\graph=(\vertices,\edges)$ with vertex set $\vertices$, edge set $\edges$, an origin $\source\in \vertices$, and a destination $\sink\in \vertices$. 
The traffic demand is given by a positive real number $\rate> 0$, interpreted as vehicles per hour, which has to be routed along a set $\routes$ of simple paths from $\source$ to $\sink$. 
The nonatomic hypothesis means that each vehicle controls a negligible fraction of the total traffic, and consequently the traffic flows are treated as continuous variables. 

The  traffic flow on path $\route$ is denoted by $\flow_{\route}$ and the set of \emph{feasible flow profiles} is 
\begin{equation}
\label{eq:flows}
\flows_{\rate}=\braces*{\flowprof=(\flow_{\route})_{\route\in\routes} \colon \flow_{\route}\ge 0 \text{ and }\sum_{\route\in\routes} \flow_{\route}=\rate}.
\end{equation}
Each flow profile $\flowprof\in \flows_{\rate}$ induces a \emph{load profile} $\loadprof=(\load_{\edge})_{\edge\in\edges}$ with $\load_{\edge}=\sum_{\route\ni\edge} \flow_{\route}$ 
representing the aggregate traffic over the edge $\edge$. 
We call $\loads_{\rate}$ the set of all such  load profiles.
Note that different flow profiles may induce the same  edge loads  so this correspondence is not  bijective.

Every edge $\edge\in\edges$ has a continuous nondecreasing \emph{cost function} $\cost_{\edge} \colon [0,+\infty)\to [0,+\infty)$, where
$\cost_{\edge}(\load_{\edge})$ represents the travel time
(or unit cost) of traversing the edge when the load is $\load_{\edge}$. 
When the traffic is distributed according to $\flowprof=(\flow_{\route})_{\route\in\routes}$ 
with induced load profile $\loadprof=(\load_{\edge})_{\edge\in\edges}$, the cost experienced by traveling on a path $\route\in\routes$
is given by 
\begin{equation}
\label{eq:cost-path}
\cost_{\route}(\flowprof)
\coloneqq\sum_{\edge\in\route} \cost_{\edge}\parens*{\sum_{\routealt\ni\edge}\flow_{\routealt}}
=\sum_{\edge\in\route} \cost_{\edge}(\load_{\edge}).
\end{equation}
With a slight abuse of notation, we use the same symbol $\cost$ for the cost function over paths and over edges. 
The meaning should be clear from the context.

\subsection{Wardrop equilibrium}
\label{suse:Wardrop}

A feasible flow profile  is called a \emph{Wardrop equilibrium} if the paths that are actually used have minimum cost. 
Formally, $\eq{\flowprof}\in \flows_{\rate}$ is an equilibrium iff there exists  $\eqcost\in\R$ such that
\begin{equation}\label{eq:Wardrop}
\begin{cases}
\cost_{\route}(\eq{\flowprof})=\eqcost & \text{for each }\route\in\routes\text{ such that }\eq{\flow}_{\route}>0,\\
\cost_{\route}(\eq{\flowprof})\ge\eqcost & \text{for each }\route\in\routes\text{ such that }\eq{\flow}_{\route}=0.
\end{cases}
\end{equation}
The quantity $\eqcost$ is called the \emph{equilibrium cost}, and is a function of $\rate$.

As noted by \citet{BecMcGWin:Yale1956}, Wardrop equilibria coincide with the optimal solutions of the convex minimization problem
\begin{equation}
\label{eq:Wardrop-variational}
\valueW(\rate)\coloneqq \min_{\flowprof\in \flows_{\rate}} \sum_{\edge\in\edges} \Cost_{\edge}\parens*{\sum_{\routealt\ni\edge}\flow_{\routealt}}
=\min_{\load\in \loads_{\rate}}\sum_{\edge\in\edges}\Cost_{\edge}(\load_{\edge}),
\end{equation}
where $\Cost_{\edge}(\argdot)$ is the primitive of the edge cost $\cost_{\edge}(\argdot)$, that is,
\begin{equation}
\label{eq:Cost}
\Cost_{\edge}(\load_{\edge})=\int_{0}^{\load_{\edge}} \!\!\cost_{\edge}(\zvar)\diff \zvar.
\end{equation}
This follows by noting that \eqref{eq:Wardrop} are the optimality conditions for $\valueW(\rate)$, with the equilibrium cost $\eqcost$ playing the role of a Lagrange multiplier for the constraint $\sum_{\route} \flow_{\route}=\rate$.
It follows that, for each fixed demand $\rate$ an equilibrium flow $\eq{\flowprof}$ always exists.  

Although Wardrop equilibria are not always unique, all of them induce the same edge costs
$\eqcostedge_{\edge}(\rate)\coloneqq\cost_{\edge}(\eq{\load}_{\edge})$. 
In particular they have the same equilibrium cost $\eqcost=\eqcost(\rate)$, which is simply the shortest $\source$-$\sink$ distance:
\begin{equation}\label{eq:lambda-min}
\eqcost(\rate)=\min_{\route\in\routes}\sum_{\edge\in\route}\eqcostedge_{\edge}(\rate).
\end{equation}
As a matter of fact, as proved by \citet{Fuk:TRB1984}, the equilibrium edge costs $\eqcostedge_{\edge}=\eqcostedge_{\edge}(\rate)$ 
are the unique optimal solution of the strictly convex dual program
\begin{equation}\label{eq:lambda-dual}
\min_{\eqcostedge}\sum_{\edge\in\edges}\Cost_\edge^*(\eqcostedge_{\edge}) - \rate\,\min_{\route\in\routes}\sum_{\edge\in\route}\eqcostedge_{\edge}
\end{equation}
where $\Cost_\edge^*(\argdot)$ is the Fenchel conjugate of $\Cost_\edge(\argdot)$, which is strictly convex.

\subsection{Social optimum and efficiency of equilibria}
The total cost experienced by all users traveling across the network is called the \emph{social cost} and is denoted by
\begin{equation}\label{eq:SC}
\SC(\flowprof)\coloneqq \sum_{\route\in\routes}\flow_{\route}\, \cost_{\route}(\flowprof)=\sum_{\edge\in\edges}\load_{\edge}\, \cost_{\edge}(\load_{\edge}).
\end{equation}
Since in equilibrium all the paths that carry flow have the same cost $\eqcost(\rate)$, it follows that all equilibria have the same social cost, that is,
\begin{equation}\label{eq:same-cost}
\SC(\eq{\flowprof})=\sum_{\route\in\routes}\eq{\flow}_{\route}\,\cost_{\route}(\eq{\flowprof})
=\sum_{\route\in\routes}\eq{\flow}_{\route}\,\eqcost(\rate) = \rate\,\eqcost (\rate).
\end{equation}

A feasible flow $\opt{\flowprof}\in \flows_{\rate}$ is called an \emph{optimum flow} if it minimizes the social cost, that is, $\opt{\flowprof}$ is an optimal solution of
\begin{equation}
\label{eq:opt-min}
\opt{\valueW}(\rate)\coloneqq \min_{\flowprof\in \flows_{\rate}} \SC(\flowprof)
=\min_{\flowprof\in \flows_{\rate}} \sum_{\edge\in\edges}\opt{\Cost}_{\edge}\parens*{\sum_{\routealt\ni\edge}\flow_{\routealt}},
\end{equation}
where $\opt{\Cost}_{\edge}(\load_{\edge})=\load_{\edge}\, \cost_{\edge}(\load_{\edge})$. 
The \acfi{PoA}\acused{PoA} is then defined 
as the ratio between the social cost at equilibrium $ \rate\,\eqcost(\rate)$ and the minimum social cost $\opt{\valueW}(\rate)$: 
\begin{equation}\label{eq:PoA}
\PoA(\rate)= \frac{ \rate\,\eqcost(\rate)}{\opt{\valueW}(\rate)}.
\end{equation}

Our main goal is to investigate the smoothness of the function $\rate\mapsto\PoA(\rate)$ and to understand  
the kinks and monotonicity properties observed in the example of \cref{fi:oscillations-b}. To this end, \cref{sec:prelim} presents
some preliminary results on the differentiability of equilibria as a function of the demand $\rate$.
More precise results will be discussed in \cref{se:affine} for the case with affine costs.

\section{Differentiability of equilibria and \acl{PoA}}
\label{sec:prelim}
 In order to study the smoothness of the PoA, we begin by establishing some preliminary facts on the differentiability of the 
 optimal value function $\rate\mapsto \valueW(\rate)$ 
and the smoothness of the equilibrium loads $\rate\mapsto\eq{\load}_{\edge}(\rate)$. These results
follow from general convex duality and sensitivity analysis of parametric optimization problems.
The following property does not seem to have been stated earlier in the literature, at least in this generality.

\begin{proposition} 
\label{pr:V-C1} 
The map $\rate\mapsto \valueW(\rate)$ is convex and $C^1$ on $(0,\infty)$ with $\valueW'(\rate)=\eqcost(\rate)$ 
continuous and nondecreasing.
Moreover, the equilibrium costs $\eqcostedge_{\edge}(\rate)$ are uniquely defined and continuous.
\end{proposition}

\begin{proof}
This is a consequence of the convex duality theorem. 
Indeed, consider the 
perturbation function $\perturb_{\rate}: \R^{|\routes|}\times \R\to\R\cup\{+\infty\}$, given by
\begin{equation}\label{eq:perturb}
\perturb_{\rate}(\flowprof,\zvar)
=
\begin{cases}
\sum_{\edge\in\edges}\cost_{\edge}\parens*{\sum_{\routealt\ni\edge}\flow_{\routealt}} & \text{if } \flowprof\geq \zerovec,\ \sum_{\route}\flow_{\route}=\rate+\zvar\\
+\infty& \text{otherwise}.
\end{cases}
\end{equation}
Clearly $\perturb_{\rate}$ is a proper closed convex function \citep[page~24]{Roc:PUP1997}.
Calling 
\begin{equation}\label{eq:v-m}
\vinf_{\rate}(\zvar)=\inf_{\flowprof}\perturb_{\rate}(\flowprof,\zvar),
\end{equation}
we have
$\valueW(\rate+\zvar)=\vinf_{\rate}(\zvar)$ and in particular $\valueW(\rate)=\vinf_{\rate}(0)$ which we consider as the primal problem $(\prim_{\rate})$.
From general convex duality, we have that $\zvar\mapsto \vinf_{\rate}(\zvar)=\valueW(\rate+\zvar)$ is a convex function,  from which we deduce that $\valueW(\argdot)$ is convex.
Moreover, the perturbation function $\perturb_{\rate}$ yields a corresponding dual 
\begin{equation}\label{eq:dual}
(\dual_{\rate})\quad\min_{\eqcost\in\R}\perturb_{\rate}^{*}(\zerovec,\eqcost),
\end{equation}
where $\perturb_{\rate}^{*}$ is the Fenchel conjugate function, that is,
\begin{align*}
\perturb_{\rate}^{*}(\zerovec,\eqcost)
&=\sup_{\flowprof,\zvar}~\langle \zerovec,\flowprof\rangle+\eqcost \zvar-\perturb_{\rate}(\flowprof,\zvar)\\
&=\sup_{\flowprof\geq \zerovec}\eqcost\parens*{\sum_{\route} \flow_{\route}-\rate}-\sum_{\edge\in\edges}\Cost_{\edge}\parens*{\sum_{\routealt\ni\edge}\flow_{\routealt}}.
\end{align*}

Since $\valueW(\rate')$ is finite for all $\rate'\in (0,\infty)$, it follows that $\vinf_{\rate}(\zvar)=\valueW(\rate+\zvar)$ is finite for $\zvar$ in  some interval around $0$,
and then the convex duality theorem implies that there is no duality gap and the subdifferential $\partial \vinf_{\rate}(0)$ at $\zvar=0$ coincides with the optimal solution 
set $\solset(\dual_{\rate})$ of the dual problem, that is, $\partial \valueW(\rate)=\partial \vinf_{\rate}(0)=\solset(\dual_{\rate})$. 

We claim that the dual problem has a unique solution, which is exactly the equilibrium cost $\eqcost(\rate)$. 
Indeed,  fix an optimal solution $\eq{\flowprof}$ for $\vinf_{\rate}(0)
=\valueW(\rate)$ and recall that this is just a Wardrop equilibrium. 
The dual optimal solutions  are precisely the $\eqcost$'s such that 
\begin{equation*}
\perturb_{\rate}(\eq{\flowprof},0)+\perturb_{\rate}^{*}(\zerovec,\eqcost)=0.
\end{equation*}
This equation can be written explicitly as
\begin{equation*}
\sum_{\edge\in\edges}\Cost_{\edge}\parens*{\sum_{\route\ni\edge}\eq{\flow}_{\route}}+\sup_{\flowprof\geq \zerovec}~\eqcost\parens*{\sum_{\route} \flow_{\route}-\rate}-\sum_{\edge\in\edges}\Cost_{\edge}\parens*{\sum_{\routealt\ni\edge}\flow_{\routealt}}=0,
\end{equation*}
from which it follows that $\flowprof=\eq{\flowprof}$ is an optimal solution in the  latter supremum. The corresponding optimality conditions are
\begin{align*}
\eqcost-\sum_{\edge\in\route}\cost_{\edge}\parens*{\sum_{\routealt\ni\edge}
\eq{\flow}_{\routealt}}=0,&\quad\text{ if }\eq{\flow}_{\route}>0,\\
\eqcost-\sum_{\edge\in\route}\cost_{\edge}\parens*{\sum_{\routealt\ni\edge}
\eq{\flow}_{\routealt}}\leq 0,&\quad\text{ if }\eq{\flow}_{\route}=0,
\end{align*}
which imply that $\eqcost$ is the equilibrium cost for the Wardrop equilibrium, that is, $\eqcost=\eqcost(\rate)$. It follows that $\partial \valueW(\rate)=\{\eqcost(\rate)\}$ so that $\rate\mapsto \valueW(\rate)$ is not only
convex but also differentiable with $\valueW'(\rate)=\eqcost(\rate)$. The conclusion follows by noting that every convex differentiable 
function is automatically of class $C^{1}$, with $\valueW'(\rate)$ nondecreasing.

The continuity of the equilibrium edge costs $\eqcostedge_{\edge}=\eqcostedge_{\edge}(\rate)$ is a consequence of Berge's maximum theorem \citep[see, e.g.,][Section~17.5]{AliBor:Springer2006}.
Indeed, the equilibrium edge costs are optimal solutions for the dual program in \cref{eq:lambda-dual}. 
Since the objective function is jointly continuous in $(\eqcostedge,\rate)$, Berge's theorem implies that the optimal solution correspondence is upper-semicontinous. However, in this case the optimal solution is unique, so that the optimal correspondence is single-valued, and, as a consequence, the equilibrium edge costs $\eqcostedge_{\edge}(\rate)$ are continuous. 
\end{proof}

As an immediate consequence of \cref{pr:V-C1} we obtain the following result:

\begin{corollary} 
\label{cor:V-C1} 
If the costs $\cost_{\edge}(\argdot)$ are strictly increasing and continuous,  the equilibrium loads $\eq{\load}_{\edge}(\rate)=\cost_{\edge}^{-1}(\eqcostedge_{\edge}(\rate))$ 
are unique and continuous.
\end{corollary}

For multiple origin-destination networks with continuous and strictly increasing costs, the continuity of 
the \emph{equilibrium loads} $\eq{\load}_\edge$  as a function of the demands was already proved in \citet[Theorem 1]{Hal:TS1978}. 
On the other hand, \citet[Theorem 2]{Hal:TS1978} proved the continuity of the \emph{equilibrium costs} provided that all paths carry a strictly positive flow, 
while \citet[Theorem 3]{Hal:TS1978} showed that the equilibrium cost of each \ac{OD} increases with the corresponding demand. 
As shown  in \cref{pr:V-C1}, for a single origin-destination the continuity and monotonicity of the equilibrium cost requires neither that costs be strictly increasing  nor that all paths carry a strictly positive flow.
Although this might be considered a minor improvement, allowing for nondecreasing costs and particularly constant costs is a 
convenient extension. 
 
For the analysis of the \ac{PoA}, the most relevant part of \cref{pr:V-C1} is the smoothness of $\valueW(\argdot)$ and the characterization of its derivative $\valueW'(\rate)=\eqcost(\rate)$.
In particular, considering the social optimum problem \eqref{eq:opt-min} we get the following direct consequence:
\begin{proposition}\label{pr:PoA-differentiable}
Let the costs $\cost_{\edge}(\argdot)$ be $C^{1}$ and nondecreasing with $\load_{\edge}\mapsto \load_{\edge}\,\cost_{\edge}(\load_{\edge})$ convex. 
Then the optimal social cost $\rate\mapsto\opt{\valueW}(\rate)$ is convex and $C^{1}$.  
Moreover, $\rate\mapsto\PoA(\rate)$ 
is continuous in $(0,+\infty)$ and differentiable at every $\rate$ where the equilibrium cost $\eqcost(\rate)$ is differentiable. 
\end{proposition}

\begin{proof}
The assumptions on $\cost_{\edge}(\argdot)$ imply that the \emph{marginal costs}
\begin{equation}
\label{eq:marginal-cost}
\opt{\cost}_{\edge}(\load_{\edge})\coloneqq \opt{\Cost}_{\edge}'(\load_{\edge})
=\cost_{\edge}(\load_{\edge})+\load_{\edge} \,\cost_{\edge}'(\load_{\edge})
\end{equation}
are continuous and nondecreasing. It follows that the optimal flows are the Wardrop equilibria for these marginal costs,
and the smoothness of $\opt{\valueW}(\rate)$ follows from \cref{pr:V-C1}. 
For the \ac{PoA} it suffices to observe  that $\eqcost(\rate)$ is continuous and then use the equality \eqref{eq:PoA}. 
\end{proof}

\citet{WuMoh:arXiv2020}  recently established a very general result on the continuity of the \ac{PoA} with respect 
to the demands and also with respect to perturbations of the cost functions. However, differentiability was not addressed in their paper.

\subsection{Differentiability of equilibrium costs}
In order to use \cref{pr:PoA-differentiable} it is convenient to find conditions that ensure the 
differentiability of the equilibrium cost $\eqcost(\rate)$. In this section we present one such result, which 
also guarantees the differentiability of the resource loads $\eq{\load}_{\edge}(\rate)$.
This follows from the implicit function theorem applied to the  
system of first order optimality conditions for \cref{eq:Wardrop-variational}.
Given a vertex $\vertex\in\vertices$, call $\outneigh(\vertex)$ and $\inneigh(\vertex)$ the sets of out-edges and in-edges of $\vertex$, respectively, and $\routes(\source,\vertex)$ the set of all paths from $\source$ to $\vertex$. 
Moreover, call $\edges\parens*{\vertex,\vertexalt}$ the set of all edges $\edge\in\edges$ that go from $\vertex$ to $\vertexalt$.
Then, an equilibrium load profile $\parens{\eq{\load}_{\edge}}_{\edge\in\edges}$ for a total demand $\rate$ is characterized as a solution of 
\begin{align}
\label{eq:Nonnegativity}
&\load_{\edge}\ge 0 ~\text{ for all }~\edge\in\edges,\\
\label{eq:FlowConservation}
&\sum_{\edge\in\outneigh(\vertex)}\load_{\edge}-\sum_{\edge\in\inneigh(\vertex)}\load_{\edge}=
\begin{cases}
~~0 & \text{ if }\vertex\ne\source,\sink,\\
~~\rate & \text{ if }\vertex=\source,\\
-\rate & \text{ if }\vertex=\sink,
\end{cases}
\quad\text{for all }\vertex\in\vertices, \\
\label{eq:BestResponse}
&\eqcostvertex_{\source}=0~\text{ and }~\eqcostvertex_{\vertexalt}\le \eqcostvertex_{\vertex} +\eqcostedge_{\edge}~\text{ for all }~\edge\in\edges(\vertex,\vertexalt),\\
\label{eq:Complementarity}
&\load_{\edge}\cdot (\eqcostvertex_{\vertexalt}-\eqcostvertex_{\vertex}-\eqcostedge_{\edge})=0~\text{ for all }~\edge\in\edges(\vertex,\vertexalt),\\
\label{eq:tau-is-cost}
&\eqcostedge_{\edge}=\cost_{\edge}(\load_{\edge})~\text{ for all }~\edge\in\edges,
\end{align}
where $\eqcostedge_{\edge}$ is the equilibrium cost of the edge $\edge$ and 
\begin{equation}
\label{eq:T-min}
\eqcostvertex_{\vertex}\coloneqq
\min_{\route\in\routes(\source,\vertex)}\sum_{\edge\in\route}\eqcostedge_{\edge}
\end{equation}
is the equilibrium cost of a shortest path from the origin $\source$ to vertex $\vertex$.

For the subsequent analysis we define the \emph{active network} as the set of all edges that  lie
on some shortest path. We also consider the demand levels at which this set changes.
 
\begin{definition} 
\label{de:activenetwork}
For each $\rate>0$ we let $\routes(\rate)$ be the set of all shortest paths from the source $\source$ to the sink $\sink$
with cost at equilibrium equal to $\eqcost(\rate)$, and we define the \emph{active network} $\activ{\edges}(\rate)$ 
as the union of  the edges on all these paths $\route\in \routes(\rate)$.

The active network is said to be \emph{locally constant} at $\rate_{0}$ if there exists $\varepsilon>0$ such that $\activ{\edges}(\rate)$ 
is the same for all $\rate\in[\rate_{0}-\varepsilon,\rate_{0}+\varepsilon]$.

The demand $\rate_{0}$ is called an \emph{$\activ{\edges}$-breakpoint} if there exists $\varepsilon>0$ such that $\activ{\edges}(\rate)$  is constant over each of the intervals $[\rate_{0}\!-\!\varepsilon,\rate_{0})$ and $(\rate_{0},\rate_{0}\!+\!\varepsilon]$, with
$\activ{\edges}(\rate_{0}\!-\!\varepsilon)\ne\activ{\edges}(\rate_{0}\!+\!\varepsilon)$.

\end{definition}

\begin{remark}
\label{re:edge-inactive}
Since the equilibrium costs  $\eqcostedge_{\edge}=\eqcostedge_{\edge}(\rate)$ and $\eqcost(\rate)$ are unique for each $\rate$, 
it follows that both $\routes(\rate)$ and $\activ{\edges}(\rate)$ 
are also uniquely determined. 
Moreover, the continuity of $\eqcostedge_{\edge}(\rate)$ 
implies that an edge $\edge\not\in \activ{\edges}(\rate_{0})$ that is inactive at $\rate_{0}$ remains inactive for $\rate$ 
near $\rate_{0}$, that is to say $\activ{\edges}(\rate_{0})^c\subseteq \activ{\edges}(\rate)^c$.
\end{remark}

\cref{fi:active_networks} shows the evolution of the active network at different demand
levels for the game in \cref{fi:oscillations-a}, with five $\activ{\edges}$-breakpoints at $\rate=1,2,3,4,7$. 
Notice the correspondence with the break points in the \acl{PoA} in \cref{fi:oscillations-b}.
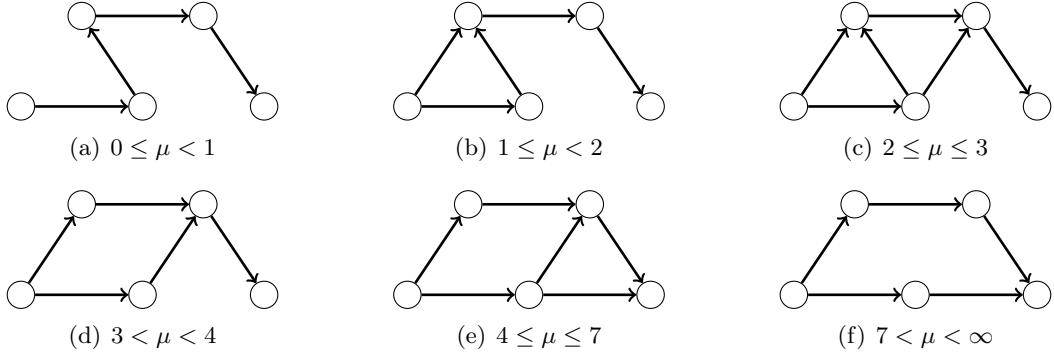
\begin{figure}[ht]
\subfigure[$0\leq\rate< 1$]  
{  
\begin{tikzpicture}[scale=.4]  
   \node[shape=circle,draw=black] (O) at (-4,0)  { }; 
   \node[shape=circle,draw=black] (v1) at (0,0)  {}; 

   \node[shape=circle,draw=black] (v2) at (-2,3)  {}; 
   \node[shape=circle,draw=black] (v3) at (2,3)  {}; 
   \node[shape=circle,draw=black] (D) at (4,0)  {}; 
    
     \draw[line width=1pt,->] (O) to  (v1);
   \draw[line width=1pt,->] (v1) to (v2);
   \draw[line width=1pt,->] (v2) to (v3);
   \draw[line width=1pt,->] (v3) to (D);

\end{tikzpicture}
}  
\hspace{1cm}
\subfigure[$1\leq\rate< 2$]  
{  
\begin{tikzpicture}[scale=.4]  
   \node[shape=circle,draw=black] (O) at (-4,0)  { }; 
   \node[shape=circle,draw=black] (v1) at (0,0)  {}; 

   \node[shape=circle,draw=black] (v2) at (-2,3)  {}; 
   \node[shape=circle,draw=black] (v3) at (2,3)  {}; 
   \node[shape=circle,draw=black] (D) at (4,0)  {}; 
    
   \draw[line width=1pt,->] (O) to  (v1);
   \draw[line width=1pt,->] (O) to  (v2);
   \draw[line width=1pt,->] (v1) to (v2);
   \draw[line width=1pt,->] (v2) to (v3);
   \draw[line width=1pt,->] (v3) to (D);

\end{tikzpicture}
}  
\hspace{1cm}
\subfigure[$2\leq\rate\leq 3$]  
{  
\begin{tikzpicture}[scale=.4]  
   \node[shape=circle,draw=black] (O) at (-4,0)  { }; 
   \node[shape=circle,draw=black] (v1) at (0,0)  {}; 

   \node[shape=circle,draw=black] (v2) at (-2,3)  {}; 
   \node[shape=circle,draw=black] (v3) at (2,3)  {}; 
   \node[shape=circle,draw=black] (D) at (4,0)  {}; 
    
   \draw[line width=1pt,->] (O) to  (v1);
   \draw[line width=1pt,->] (O) to  (v2);
   \draw[line width=1pt,->] (v1) to (v2);
   \draw[line width=1pt,->] (v1) to (v3);
   \draw[line width=1pt,->] (v2) to (v3);
   \draw[line width=1pt,->] (v3) to (D);

\end{tikzpicture}
}  

\subfigure[$3<\rate<4$]  
{  
\begin{tikzpicture}[scale=.4]  
   \node[shape=circle,draw=black] (O) at (-4,0)  { }; 
   \node[shape=circle,draw=black] (v1) at (0,0)  {}; 

   \node[shape=circle,draw=black] (v2) at (-2,3)  {}; 
   \node[shape=circle,draw=black] (v3) at (2,3)  {}; 
   \node[shape=circle,draw=black] (D) at (4,0)  {}; 
    
   \draw[line width=1pt,->] (O) to  (v1);
   \draw[line width=1pt,->] (O) to  (v2);
   \draw[line width=1pt,->] (v1) to (v3);
   \draw[line width=1pt,->] (v2) to (v3);
   \draw[line width=1pt,->] (v3) to (D);

\end{tikzpicture}
}  
\hspace{1cm}
\subfigure[$4\leq\rate\leq 7$]  
{  
\begin{tikzpicture}[scale=.4]  
   \node[shape=circle,draw=black] (O) at (-4,0)  { }; 
   \node[shape=circle,draw=black] (v1) at (0,0)  {}; 

   \node[shape=circle,draw=black] (v2) at (-2,3)  {}; 
   \node[shape=circle,draw=black] (v3) at (2,3)  {}; 
   \node[shape=circle,draw=black] (D) at (4,0)  {}; 
    
  \draw[line width=1pt,->] (O) to  (v1);
   \draw[line width=1pt,->] (O) to  (v2);
   \draw[line width=1pt,->] (v1) to (D);
   \draw[line width=1pt,->] (v1) to (v3);
   \draw[line width=1pt,->] (v2) to (v3);
   \draw[line width=1pt,->] (v3) to (D);

\end{tikzpicture}
}  
\hspace{1cm}
\subfigure[$7<\rate< \infty$]  
{  
\begin{tikzpicture}[scale=.4]  
   \node[shape=circle,draw=black] (O) at (-4,0)  { }; 
   \node[shape=circle,draw=black] (v1) at (0,0)  {}; 

   \node[shape=circle,draw=black] (v2) at (-2,3)  {}; 
   \node[shape=circle,draw=black] (v3) at (2,3)  {}; 
   \node[shape=circle,draw=black] (D) at (4,0)  {}; 
    
   \draw[line width=1pt,->] (O) to  (v1);
   \draw[line width=1pt,->] (O) to  (v2);
   \draw[line width=1pt,->] (v1) to (D);
   \draw[line width=1pt,->] (v2) to (v3);
   \draw[line width=1pt,->] (v3) to (D);

\end{tikzpicture}
}  

\caption{The active network for the graph on \cref{fi:oscillations}, with $\activ{\edges}$-breakpoints at $\rate=1,2,3,4,7$.}
\label{fi:active_networks}
\end{figure}

Although in general there can be infinitely many $\activ{\edges}$-breakpoints (see \cref{pr:RepeatingActiveRegimes} and \cref{re:repeat}),
their number is finite for series-parallel networks (\emph{cf.} \cref{pr:series-parallel-nondecreasing-traffic})
and also for general networks with affine costs (\emph{cf.} \cref{pr:affine-flows}).
In both cases, once an active network changes, it may never occur again at higher demand levels.

An edge carrying a strictly  positive flow at equilibrium must be on some optimal path, and, as a consequence,  
belongs to the active network. However, the converse may fail when a path becomes active 
but carries no flow. In order to prove the smoothness of the equilibrium flows we need to avoid this situation,
which leads to the following definition of a \emph{regular demand}. 

\begin{definition} A demand $\rate>0$ is called \emph{regular} if there is an equilibrium with $\load_{\edge}>0$ for all $\edge\in \activ{\edges}(\rate)$.
\end{definition}

Regularity is just \emph{strict complementarity}. 
Indeed, the complementarity condition \eqref{eq:Complementarity} imposes that for each $\edge\in\edges$ either $\load_{\edge}$ or $(\eqcostvertex_{\vertexalt}-\eqcostvertex_{\vertex}-\eqcostedge_{\edge})$ is zero, whereas strict complementarity requires exactly one of these expressions to be zero.
As shown next, when the costs are strictly increasing this implies that the active network is locally constant, so that a regular demand cannot  
be an $\activ{\edges}$-breakpoint. 
We note however that there can be nonregular demands at which the active network is still locally constant (see \cref{ex:WN}).

\begin{lemma} Suppose that the costs $\cost_{\edge}(\argdot)$ are strictly increasing and continuous. If $\rate_{0}$ is a regular demand then
the active network $\activ{\edges}(\rate)$ is locally constant at $\rate_{0}$.
\end{lemma}
\begin{proof} 
From regularity all active edges $\edge\in\activ{\edges}(\rate_{0})$ satisfy $\eq{\load}_{\edge}(\rate_{0})>0$.
By \cref{pr:V-C1} the maps $\eq{\load}_{\edge}(\argdot)$ are continuous, so these strict inequalities are preserved for $\rate$ near $\rate_{0}$, and therefore $\activ{\edges}(\rate_{0})\subseteq \activ{\edges}(\rate)$.
This, combined with \cref{re:edge-inactive}, yields $\activ{\edges}(\rate)= \activ{\edges}(\rate_{0})$ for $\rate$ close to $\rate_{0}$.
\end{proof}

We are now ready to establish the smoothness of the equilibrium.

\begin{proposition}
\label{pr:EqCost_Differentiable}
Assume that the costs $\cost_{\edge}(\argdot)$ are $C^1$ with strictly positive derivative. 
If $\rate_{0}$ is a regular demand then $\rate\mapsto(\eq{\loadprof}(\rate),\eqcostedgevec(\rate),\eqcostvertexvec(\rate))$ is continuously differentiable in a neighborhood of $\rate_{0}$.
In particular the equilibrium cost  $\rate\mapsto \eqcost(\rate)$ is $C^1$ near $\rate_{0}$.
\end{proposition}

\begin{proof}
From \cref{pr:V-C1,cor:V-C1}, the equilibrium costs $\eqcostedge_{\edge}(\rate)$
and loads $\eq{\load}_{\edge}(\rate)$ are uniquely defined and continuous in $\rate$.
Hence, the equilibrium cost $\eqcostvertex_\vertex(\rate)$ of a shortest path to any vertex $\vertex$ is also continuous.
These functions $\eq{\load}_{\edge}, \eqcostedge_{\edge}$ and $\eqcostvertex_\vertex$  
satisfy in particular \cref{eq:FlowConservation,eq:Complementarity,eq:tau-is-cost}.

Now, since $\rate_{0}$ is regular, the active network is locally constant.
Let $\activ{\edges}_{0}=\activ{\edges}(\rate_{0})$ be this active network and $\vertices_{0}$ the corresponding vertices.
Moreover, call $\activ{\edges}_{0}\parens*{\vertex,\vertexalt}$ the set of all edges $\edge\in\activ{\edges}_{0}$ that go from $\vertex$ to $\vertexalt$.
For $\rate$ near $\rate_{0}$ we have $\eq{\load}_{\edge}(\rate)=0$ for all $\edge\not\in\activ{\edges}_{0}$; hence, these functions are trivially differentiable. 
Also for $\vertexalt\not\in \vertices_{0}$ we can take $\vertex\in \vertices_{0}$ the last vertex on a shortest path from $\source$ to $\vertexalt$, so that $\eqcostvertex_\vertexalt(\rate)=\eqcostvertex_\vertex(\rate)+\Delta_{\vertex,\vertexalt}$ where $\Delta_{\vertex,\vertexalt}$ is a constant travel time from $\vertex$ to $\vertexalt$. Hence it suffices to establish the smoothness of $\eq{\load}_{\edge}(\rate)$ for $\edge\in\activ{\edges}_{0}$ and $\eqcostvertex_\vertex(\rate)$ for $\vertex\in\vertices_{0}$. 
To this end, consider  
\cref{eq:FlowConservation,eq:Complementarity,eq:tau-is-cost} restricted to the edges in $\activ{\edges}_{0}$, together with the equation $\eqcostvertex_{\source}=0$, which gives the following system:
\begin{align}
\label{eq:Act-FlowConservation}
&\sum_{\edge\in\outneigh(\vertex)\cap\activ{\edges}_{0}}\load_{\edge}-\sum_{\edge\in\inneigh(\vertex)\cap\activ{\edges}_{0}}\load_{\edge}=
\begin{cases}
0\text{ if }\vertex\ne\source,\sink,\\
\rate\text{ if }\vertex=\source,
\end{cases}
\quad\text{for all }\vertex\in\vertices_{0}, \\
\label{eq:Act-Complementarity}
&\load_{\edge}\cdot (\eqcostvertex_{\vertexalt}-\eqcostvertex_{\vertex}-\eqcostedge_{\edge})=0\quad\text{for all }\edge\in\activ{\edges}_{0}(\vertex,\vertexalt),\\
\label{eq:Act-tau-is-cost}
&\eqcostedge_{\edge}=\cost_{\edge}(\load_{\edge}),\\
\label{eq:TOzero}
&\eqcostvertex_{\source}=0.
\end{align} 

To apply the implicit function theorem to this reduced system, we must check that the associated linearized system has a unique solution.
Let $\incrload_{\edge}$, $\incrcostvertex_{\vertex}$ and $\incrcostedge_{\edge}$ be respectively the increments in the variables $\load_{\edge}$, $\eqcostvertex_\vertex$ and $\eqcostedge_{\edge}$ for each $\edge\in\activ{\edges}_{0}$ and $\vertex\in\vertices_{0}$. 
The homogeneous linear system obtained from \cref{eq:Act-FlowConservation,eq:Act-Complementarity,eq:Act-tau-is-cost,eq:TOzero} is:
\begin{align}
\label{eq:Lin-FlowConservation}
&\sum_{\edge\in\outneigh(\vertex)\cap\activ{\edges}_{0}}\incrload_{\edge}-\sum_{\edge\in\inneigh(\vertex)\cap\activ{\edges}_{0}}\incrload_{\edge}=
0
\quad\text{for all }\vertex\in\vertices_{0}, \\
\label{eq:Lin-Complementarity}
&\load_{\edge}(\incrcostvertex_{\vertexalt}-\incrcostvertex_{\vertex}-\incrcostedge_{\edge})+\incrload_{\edge}(\eqcostvertex_\vertexalt-\eqcostvertex_\vertex-\eqcostedge_{\edge})=0\quad\text{for all }\edge\in\activ{\edges}_{0}(\vertex,\vertexalt),\\
\label{eq:Lin-tau-is-cost}
&\incrcostedge_{\edge}=\cost_{\edge}'(\load_{\edge})\incrload_{\edge},\\
\label{eq:Lin:TOzero}
&\incrcostvertex_{\source}=0.
\end{align}
Strict complementarity on an active link implies that \cref{eq:Lin-Complementarity} is equivalent to
\begin{equation}\label{eq:Lin-SC}
\incrcostvertex_{\vertexalt}=\incrcostvertex_{\vertex}+\incrcostedge_{\edge}\qquad\text{for all }\edge\in\activ{\edges}_{0}(\vertex,\vertexalt),
\end{equation}
which, together with \cref{eq:Lin-tau-is-cost}, gives 
\begin{equation}\label{eq:45prime}
\incrcostvertex_{\vertexalt}=\incrcostvertex_{\vertex}+\cost'_{\edge} (\load_{\edge})\incrload_{\edge}\qquad\text{for all }\edge\in\activ{\edges}_{0}(\vertex,\vertexalt).
\end{equation}
These equations are the stationarity conditions for the strongly convex (since $\cost'_{\edge}(\load_{\edge})>0$) quadratic program
\begin{equation}
\label{eq:IFT-program}
\min_{\incrloadvec}\sum_{\edge\in\activ{\edges}_{0}}\frac 12 \cost'_{\edge}(\load_{\edge})\cdot \incrload_{\edge}^{2}\tag{$\prim$}
\end{equation}
under the constraints \eqref{eq:Lin-FlowConservation}. 
Indeed, associating a Lagrange multiplier to each of those constraints, we get the Lagrangian
\begin{equation*}
\lagrang(\incrloadvec,\incrcostvertexvec)=\sum_{\edge\in\activ{\edges}_{0}}\frac 12 \cost'_{\edge}(\load_{\edge})\incrload_{\edge}^{2}+\sum_{\vertex\in\vertices_{0}}\incrcostvertex_{\vertex}\left(\sum_{\edge\in\outneigh(\vertex)\cap\activ{\edges}_{0}}\incrload_{\edge}-\sum_{\edge\in\inneigh(\vertex)\cap\activ{\edges}_{0}}\incrload_{\edge} \right)
\end{equation*}
and the equation $\partial\lagrang/\partial \incrload_{\edge}=0$ is precisely equivalent to \cref{eq:45prime}.
Hence, every solution of \cref{eq:Lin-FlowConservation,eq:Lin-Complementarity,eq:Lin-tau-is-cost,eq:Lin:TOzero} corresponds to an optimal solution of \eqref{eq:IFT-program}. 
Since $\incrload_{\edge}=0$ for all $\edge\in\activ{\edges}_{0}$ is feasible, it is also the unique optimal solution. 
Then \cref{eq:Lin-tau-is-cost} yields $\incrcostedge_{\edge}=0$ for all $\edge\in\activ{\edges}_{0}$, and 
 from \cref{eq:Lin-SC,eq:Lin:TOzero} we also get $\incrcostvertex_{\vertex}=0$ for all $\vertex\in\vertices_{0}$.

Since the linear system  \cref{eq:Lin-FlowConservation,eq:Lin-Complementarity,eq:Lin-tau-is-cost,eq:Lin:TOzero} has only the trivial solution, the Jacobian of \cref{eq:Act-FlowConservation,eq:Act-Complementarity,eq:Act-tau-is-cost,eq:TOzero} with respect to $\loadprof,\eqcostedgevec,\eqcostvertexvec$ is invertible and the implicit function theorem implies the smoothness of the solution. In particular $\eqcost(\rate)=\eqcostvertex_\sink(\rate)$ is continuously differentiable.
\end{proof}

\begin{remark}
The directional differentiability of the equilibrium loads with respect to parameters was investigated by \citet{Pat:TS2004} and \citet{JosPat:TRB2007}, 
though it can also be derived from \citet[Theorem~5.1]{Sha:SIAMJCO1988}.
Theorem~10 in \citet{Pat:TS2004} shows that differentiability holds provided that the edge loads in the linearized 
system are unique (this is what is done in the proof above) and assuming in addition that all the route flows corresponding 
to the solutions of the linearized system satisfy an additional vanishing condition. 
Our \cref{pr:EqCost_Differentiable} is more straightforward as it follows directly from the implicit function theorem, which gives in addition the continuity of the derivatives. 
It is also easier to apply and to interpret: it just requires to check that all active links carry a positive flow. 
In other words, smoothness can only fail at critical values of the demand where a new link becomes active but is not yet carrying flow. 
This simpler sufficient condition is all that is needed hereafter.
\end{remark}

Combining \cref{pr:PoA-differentiable,pr:EqCost_Differentiable} we obtain the following result on the differentiability of the \acl{PoA}.

\begin{theorem}\label{thm:PoA-differentiable}
Suppose that $\cost_{\edge}(\argdot)$ are $C^{1}$ with strictly positive derivative and 
$\load_{\edge}\,\cost_{\edge}(\load_{\edge})$ convex. Then $\PoA(\argdot)$ is continuously differentiable at each regular demand level $\rate_{0}$.
\end{theorem}

While all $\activ{\edges}$-breakpoints are nonregular, there might exist other nonregular points 
that are not $\activ{\edges}$-breakpoints (see \cref{ex:WN}). We do not know if 
differentiability of $\PoA(\argdot)$ can fail at such additional nonregular points. 
In the next section we will show that for affine costs,
nonsmoothness can only occur at $\activ{\edges}$-breakpoints  and that there are finitely many of them.  
In contrast, for  general networks and nonlinear costs the number of $\activ{\edges}$-breakpoints can be unbounded. In this regard, it is worth noting
that for networks with a series-parallel topology (which excludes the Wheatstone network), the active network increases 
monotonically with the demand, which yields a sharp bound on the number of different active networks and $\activ{\edges}$-breakpoints 
that can occur as the demand grows from $0$ to $+\infty$.

\begin{definition}
\label{def:series-parallel}
The class of \acfi{SP}\acused{SP} networks can be constructed as follows:
\begin{itemize}
\item
A network with two vertices $\source, \sink$ and one edge $(\source,\sink)$ connecting them is \ac{SP}.

\item
A network obtained by joining in series two \ac{SP} networks by merging $\sink_{1}$ with $\source_{2}$ is \ac{SP}.

\item
A network obtained by joining in parallel two \ac{SP} networks by merging $\source_{1}$ with $\source_{2}$ and $\sink_{1}$ with $\sink_{2}$ is \ac{SP}.
\end{itemize}
\end{definition}

\begin{proposition}
\label{pr:series-parallel-nondecreasing-traffic}
Let $\graph$ be a \acl{SP} network. 
Then there exist equilibrium load profiles $\loadprof(\rate)$ whose components $\load_{\edge}(\rate)$ are nondecreasing functions of the demand $\rate$. Moreover, the active network $\activ{\edges}(\rate)$ is also nondecreasing with respect to inclusion
so that
the number of $\activ{\edges}$-breakpoints is bounded by the minimum between the number of paths and the number of edges.
\end{proposition}

\begin{proof}
The result clearly holds for the network with only two vertices and a single link. 

Let $\graph^{1}$ and $\graph^{2}$ be two  \acl{SP}  networks for which the result is true,
and fix  two nondecreasing equilibrium loads ${\loadprof}^{1}(\rate)$ and ${\loadprof}^{2}(\rate)$, with their corresponding equilibrium costs $\eqcost^{1}(\rate)$ and $\eqcost^{2}(\rate)$ and active networks $\edges^{1}(\rate)$ and $\edges^{2}(\rate)$.

If $\graph^{1}$ and $\graph^{2}$ are connected in series, then an equilibrium is given by the coupling $\loadprof(\rate)=(\loadprof^{1}(\rate),\loadprof^{2}(\rate))$ with active network $\activ{\edges}(\rate)=\edges^{1}(\rate)\cup\edges^{2}(\rate)$, all of which are nondecreasing with $\rate$. 

The case in which $\graph^{1}$ and $\graph^{2}$ are joined in parallel, is slightly more involved.
Here an equilibrium splits as ${\loadprof}(\rate)=(\loadprof^{1}(\rate^{1}),\loadprof^{2}(\rate^{2}))$ where $\rate^{1}+\rate^{2}=\rate$ with $\rate^{2}=0$ if $\eqcost^{1}(\rate)<\eqcost^{2}(0)$, $\rate^{1}=0$ if $\eqcost^{2}(\rate)<\eqcost^{1}(0)$, and $\eqcost^{1}(\rate^{1})=\eqcost^{2}(\rate^{2})$ otherwise. 
More explicitly, if we let $g(\rate)=\inf\{x\in[0,\rate]: \eqcost^{1}(x)\geq \eqcost^{2}(\rate-x)\}$, with $g(\rate)=\rate$ when the latter set is empty, then both $\rate^{1}=g(\rate)$ and $\rate^{2}=\rate-g(\rate)$ turn out to be nondecreasing and therefore $\loadprof(\rate)=(\loadprof^{1}(g(\rate)),\loadprof^{2}(\rate-g(\rate)))$ is a nondecreasing equilibrium.
The monotonicity of the active network $\activ{\edges}(\rate)$ is similar. 
We have $\activ{\edges}(\rate)=\edges^{1}(\rate)$ if $\eqcost^{1}(\rate)<\eqcost^{2}(0)$,
$\activ{\edges}(\rate)=\edges^{2}(\rate)$ when $\eqcost^{2}(\rate)<\eqcost^{1}(0)$,
and $\activ{\edges}(\rate)=\edges^{1}(g(\rate))\cup\edges^{2}(\rate-g(\rate))$ otherwise,
and in all three cases the active network is nondecreasing.
\end{proof}

\begin{remark} 
A related result was obtained by \citet{Mil:GEB2006} for undirected networks.
His Lemma~2 shows that in a series-parallel network \emph{there exists some path} whose edge loads are increasing in the total traffic demand. 
\cref{pr:series-parallel-nondecreasing-traffic} proves that there exists an equilibrium in which this monotonicity holds for all edges and paths.

Concerned about complexity of computing parametric mincost flows, \citet[Corollary~4]{KliWar:MORfrth} also proved monotonicity of the output flows on the edges when costs are piecewise linear. 

\end{remark}

%----------------------------------------------------------------------
%%% AFFINE COSTS
%----------------------------------------------------------------------

\section{Networks with affine cost functions.}
\label{se:affine}

\subsection{Equilibrium flows}
\label{suse:affine-equilibrium-flows}

In this section we consider the case of affine cost  functions 
\begin{equation}
\label{eq:affine-costs}
\cost_{\edge}(\load)
=\coeffa_{\edge}\cdot \load+\coeffb_{\edge},
\end{equation}
with $\coeffa_{\edge},\coeffb_{\edge}\ge 0$ for each $\edge\in\edges$. 
We recall that in this case we have a scaling law that relates the equilibrium and  optimum flows.

\begin{lemma}
[\protect{\citet[Lemma~2.3]{RouTar:JACM2002}}]
\label{pr:scaling-law}
Suppose $\cost_{\edge}$ is affine for all $\edge\in\edges$. 
Let $\eq{\flowprof}(\rate)$ be an equilibrium flow with corresponding load $\eq{\loadprof}(\rate)$.
Then $\opt{\flowprof}(\rate)=\frac{1}{2}\eq{\flowprof}(2\rate)$ is a socially optimal flow with load
$\opt{\loadprof}(\rate)=\frac{1}{2}\eq{\loadprof}(2\rate)$.
\end{lemma}

From this it follows directly that for affine costs the $\activ{\edges}$-breakpoints for the optimum are in one-to-one correspondence with the $\activ{\edges}$-breakpoints for the equilibrium, that is,
\begin{equation}
\label{eq:opt-m}
\opt{\rate}_{0}=\frac{1}{2}\rate_{0}.
\end{equation}

\citet[Theorem~3.5]{OHaConWat:TRB2016} established a similar scaling law when all the edge costs are $\BPR$ functions of the same degree.

We now describe the behavior of the equilibrium and the \acl{PoA}.
We first show that any given subset of edges $\activ{\edges}_{0}\subseteq\edges$ can only be an active network over an interval. In other words,
once a given active network changes it will never become active again. This implies that the number of $\activ{\edges}$-breakpoints is finite. 
Moreover, we  show that between $\activ{\edges}$-breakpoints the equilibrium cost is affine  with nonnegative slope and intersect.

\begin{proposition}
\label{pr:affine-flows}
Suppose that the costs $\cost_{\edge}(\argdot)$ are affine.
Let $\edges_{0}\subseteq\edges$ and suppose that $\rate^1<\rate^2$ are such that 
$\activ{\edges}(\rate^1)=\activ{\edges}(\rate^2)=\edges_{0}$. 
Then, for every $\rate\in[\rate^1,\rate^2]$ we have $\activ{\edges}(\rate)=\edges_{0}$ and
we can select an equilibrium flow $\flowprof(\rate)$ that is affine in $\rate$, so that the equilibrium cost $\eqcost(\rate)$ is also affine on $[\rate^1,\rate^2]$. 
More precisely $\eqcost(\rate)=\coeffalpha+\coeffbeta\rate$ with 
nonnegative coefficients $\coeffalpha\geq 0$ and $\coeffbeta\geq 0$.
\end{proposition}

\begin{proof} Let  $\flowprof^1$ and $\flowprof^2$ be equilibrium profiles for $\rate^1$ and $\rate^2$,
and consider the following affine interpolation with $\alphaint=\alphaint(\rate)\coloneqq(\rate-\rate^1)/(\rate^2-\rate^1)\in[0,1]$  
\begin{equation}\label{eq:afineflow}
\flowprof(\rate)=(1-\alphaint)\flowprof^1+\alphaint\flowprof^2\in\flows_{\!\rate}.
\end{equation}
The corresponding load profile is given by $\loadprof(\rate)=(1-\alphaint)\loadprof^1+\alphaint\loadprof^2$ and, since the costs are affine, it follows that $\cost_{\edge}(\load_{\edge}(\rate))=(1-\alphaint)\cost_{\edge}(\load_{\edge}^{1})+\alphaint\cost_{\edge}(\load_{\edge}^{2})$.
Hence, the same affine behavior holds for the path costs
$\cost_{\route}(\flowprof(\rate))=(1-\alphaint)\cost_{\route}(\flowprof^{1})+\alphaint\cost_{\route}(\flowprof^{2})$.
Now, since $\activ{\edges}(\rate^1)=\activ{\edges}(\rate^2)=\edges_{0}$, the optimal paths are the same for $\rate^1$ and $\rate^2$.  Thus, if $\route$ is
an optimal path and $\routealt$ is not optimal, we have 
\begin{align*}
\cost_{\route}(\flowprof^1)&=\eqcost(\rate^1)<\cost_{\routealt}(\flowprof^1)\\
\cost_{\route}(\flowprof^2)&=\eqcost(\rate^2)<\cost_{\routealt}(\flowprof^2)
\end{align*}
and taking a convex combination of these inequalities we get
\begin{equation*}
\cost_{\route}(\flowprof(\rate))=(1-\alphaint)\eqcost(\rate^1)+\alphaint \eqcost(\rate^2)<\cost_{\routealt}(\flowprof(\rate)).
\end{equation*}
This implies that the paths $\route$ and $\routealt$ remain respectively optimal and nonoptimal  for $\rate$.
It follows that $\activ{\edges}(\rate)=\edges_{0}$ and also that $\flowprof(\rate)$ is an equilibrium with 
$\eqcost(\rate)=(1-\alphaint)\eqcost(\rate^1)+\alphaint \eqcost(\rate^2)$.

This shows that the equilibrium cost is affine over the interval $[\rate^1,\rate^2]$, that is, $\eqcost(\rate)=\coeffalpha+\coeffbeta\rate$. 
From \cref{pr:V-C1} we know that $\eqcost(\rate)$ 
is nondecreasing so that $\coeffbeta\geq 0$, and therefore it remains to show that $\coeffalpha\geq 0$. 
We rewrite the interpolated flow as $\flowprof(\rate)=\rate \wvec +\zvec$ with
\begin{align}
\label{eq:w}
\wvec&=\frac{\flowprof^2-\flowprof^1}{\rate^2-\rate^1},\\
\label{eq:z}
\zvec&=\frac{\rate^2\flowprof^1-\rate^1\flowprof^2}{\rate^2-\rate^1}.
\end{align}
Let $\routes_{0}$  be the set of all the $\source$-$\sink$ shortest paths included in $\edges_{0}$.
We observe that $\zvec_{\route}=0$ for every $\route\not\in\routes_{0}$, and also that $\sum_{\route\in\routes}\zvec_{\edge}=0$.
On the other hand, since $\flowprof(\rate)$ is an equilibrium we have $\cost_{\route}(\flowprof(\rate))=\eqcost(\rate)$ for all $\route\in\routes_{0}$,
so that $\langle c(\flowprof(\rate)),\zvec\rangle=0$ and, as a consequence,
\begin{equation}
\label{eq:mu-lambda}
\rate\eqcost(\rate)=\braket*{\cost(\flowprof(\rate))}{\flowprof(\rate)}
=\braket*{\cost(\flowprof(\rate))}{\rate\wvec+\zvec}
=\rate\braket*{\cost(\flowprof(\rate))}{\wvec}.
\end{equation}
Defining 
\begin{equation}
\label{eq:A-beta}
\matrA=\epincid^{\top}\diaga\epincid \quad\text{and}\quad \sumbvec=\epincid^{\top}\coeffbprof,
\end{equation} 
with $\epincid$ the edge-path incidence matrix, $\diaga=\diag[(\coeffa_{\edge})_{\edge\in\edges}]$,
and $\coeffbprof=(\coeffb_{\edge})_{\edge\in\edges}$, 
the vector of path costs can be expressed as 
\begin{equation}
\label{eq:path-cost-vec}
\cost(\flowprof(\rate))=\matrA\flowprof(\rate)+\sumbvec=\rate\matrA\wvec+\matrA\zvec+\sumbvec,
\end{equation}
so that 
\begin{equation}
\label{eq:lambda(mu)}
\eqcost(\rate)=\braket*{\cost(\flowprof(\rate))}{\wvec}
= \rate\braket*{\matrA \wvec}{\wvec}+\braket*{\matrA \zvec\!+\!\sumbvec}{\wvec},
\end{equation}
which yields $\coeffbeta=\braket*{\matrA \wvec}{\wvec}$ and $\coeffalpha=\braket*{\matrA\zvec\!+\!\sumbvec}{\wvec}$.
Since $\matrA$ is positive semidefinite, we get again $\coeffbeta\geq 0$. 
Now, since $0=\langle \cost(\flowprof(\rate)),\zvec\rangle=\rate\langle\matrA\wvec,\zvec\rangle+\langle\matrA\zvec+\sumbvec,\zvec\rangle$ 
for all $\rate$, it follows that $\langle\matrA\wvec,\zvec\rangle=0$ so that  $\coeffalpha=\braket*{\sumbvec}{\wvec}$,
and also $\langle\matrA\zvec+\sumbvec,\zvec\rangle=0$ so that $\langle\sumbvec,\zvec\rangle=-\langle\matrA\zvec,\zvec\rangle\leq 0$.
To conclude, we note that all the entries in $\flowprof(\rate)$ and in $\sumbvec$ are nonnegative, and therefore
\begin{equation}
\label{eq:f-beta}
0\leq\braket*{\flowprof(\rate)}{\sumbvec}=\rate\braket*{\wvec}{\sumbvec}+\braket*{\zvec}{\sumbvec},
\end{equation}
from which we deduce that $\coeffalpha=\braket*{\wvec}{\sumbvec}\geq 0$ as claimed.
\end{proof}

\begin{remark}
The paper by \citet{KliWar:MORfrth} developed a homotopy method for computing the full path of Wardrop equilibria as a function of the traffic demand. 
The method is designed to work with piecewise linear costs and produces a piecewise linear path of equilibrium loads. 
The algorithm first determines the equilibrium costs and then recovers the loads by inverting the link costs. 
This requires the costs to be  strictly increasing.  
In contrast, we work directly in the space of flows so we can handle nondecreasing and constant costs. 
However, we restrict to affine costs, which is essential for \cref{pr:affine-flows}. 
Indeed, beyond the piecewise affine character of the equilibrium, the most relevant part of this result is the identification of the breakpoints as the demand levels at which the active network changes. 
As shown in \cref{pr:affine-flows}, each particular subnetwork can be active 
on a demand interval, and, on each of these intervals, the equilibrium varies linearly. 
Once an active network is abandoned, it will never occur again at higher demand levels. 
This property fails to hold for nonlinear or even piecewise linear costs, as shown in \cref{pr:RepeatingActiveRegimes}.
\end{remark}

\citet[Theorem~10]{KliWar:MORfrth} showed that, even for affine costs, the number of $\activ{\edges}$-breakpoints can be exponential in the number of paths.
In \cref{se:examples} we will use their example to show an interesting behavior of the \acl{PoA}.

%----------------------------------------------------------------------
%%% POA
%----------------------------------------------------------------------

\subsection{Behavior of the Price of Anarchy}

We now prove that the social cost at the equilibrium and at the optimum have a very similar quadratic form,
from which we deduce that between $\activ{\edges}$-breakpoints the function $\rate\mapsto\PoA(\rate)$ has a unique minimum and
its maximum must be attained at some of the $\activ{\edges}$-breakpoints.

\begin{proposition}
\label{pr:SocialCostAffine}
Let $\rate_{\run}$ and $\rate_{\run+1}$ be two consecutive  $\activ{\edges}$-breakpoints for the equilibrium. 
Then, there exist $\coeffalpha_{\run}\geq 0$, $\coeffbeta_{\run}\ge 0$, and $\coeffgamma_{\run}\le 0$,
such that
\begin{align}
\label{eq:SC-eq-mu}
\SC(\eq{\flowprof}(\rate))&=\coeffalpha_{\run}\rate+\coeffbeta_{\run}\rate^{2} \quad\text{when }\rate\in(\rate_{\run},\rate_{\run+1}),\\ 
\label{eq:SC-opt-mu}
\SC(\opt{\flow}(\rate))&=\coeffgamma_{\run}+\coeffalpha_{\run}\rate+\coeffbeta_{\run}\rate^{2} \quad\text{ when }2\rate\in (\rate_{\run},\rate_{\run+1}),\text{ i.e., }\rate\in (\opt{\rate}_{\run},\opt{\rate}_{\run+1}).
\end{align}
\end{proposition}

\begin{proof}
Since $\SC(\eq{\flowprof}(\rate))=\rate\,\eqcost(\rate)$, the equality \eqref{eq:SC-eq-mu} follows directly from \cref{pr:affine-flows}
with $\coeffalpha_{\run}=\braket*{\matrA\zvec\!+\!\sumbvec}{\wvec}$ and $\coeffbeta_{\run}=\braket*{\matrA \wvec}{\wvec}$,
where $\wvec, \zvec, \matrA, \sumbvec$ are defined as in \cref{eq:w,eq:z,eq:A-beta}.

In order to prove \eqref{eq:SC-opt-mu}, let $\flowprof(\rate)=\rate \wvec +\zvec$ be the affine interpolated equilibria as in the proof of \cref{pr:affine-flows}.
When $2\rate\in(\rate_{\run},\rate_{\run+1})$, \cref{pr:scaling-law} implies that an optimum flow is 
\begin{equation*}
\opt{\flowprof}(\rate)=\text{$\frac{1}{2}$}\flowprof(2\rate)=\rate\,\wvec +\text{$\frac{1}{2}$}\zvec,
\end{equation*}
Then the social cost at optimum will be a quadratic function in $\rate$ with  
the same linear coefficient ($\coeffalpha_{\run}$) and quadratic coefficient ($\coeffbeta_{\run}$)
of the social cost at equilibrium, and constant coefficient
\begin{equation*}
\coeffgamma_{\run}=\frac 14\langle \matrA\zvec,\zvec\rangle+ \frac 12\langle \sumbvec,\zvec\rangle,
\end{equation*}
which is less or equal to $\langle \matrA\zvec,\zvec\rangle+ \langle \sumbvec,\zvec\rangle=0$ since $\langle \matrA\zvec,\zvec\rangle\ge 0$ and $\langle \sumbvec,\zvec\rangle\le 0$.
\end{proof}

From this result it follows that the \ac{PoA} between $\activ{\edges}$-breakpoints is a quotient of quadratics.
Moreover, the specific signs of the coefficients of these quadratics imply that  \ac{PoA} has a unique minimum between  
breakpoints and that its local maxima can only occur at these breakpoints.

\begin{theorem}
\label{pr:UnimodalityAffine}
Let $\rate_{\run}$ and $\rate_{\run+1}$ be two consecutive  $\activ{\edges}$-breakpoints for the equilibrium. 
Then on the interval $(\rate_{\run},\rate_{\run+1})$ the function $\rate\mapsto\PoA(\rate)$ is smooth and it is 
either decreasing, or increasing, or  first decreasing and then increasing with a local minimum in the interior of the interval. In particular $\PoA(\rate)$ does not attain a local maximum 
on $(\rate_{\run},\rate_{\run+1})$.
\end{theorem}

\begin{proof}
From \cref{pr:PoA-differentiable} the optimal social cost is a $C^1$ function of $\rate$, so that
\cref{eq:SC-eq-mu} implies that $\PoA(\rate)$ is smooth over the full interval $(\rate_{\run},\rate_{\run+1})$. 
Consider first the case when there is no $\activ{\edges}$-breakpoint $\opt{\rate}_{\ell}$ for the optimum in this interval.
Using \cref{eq:SC-eq-mu,eq:SC-opt-mu} we can express the \ac{PoA} in the form
\begin{equation*}
\PoA(\rate)
=\frac{\SC(\eq{\flowprof}(\rate))}{\SC(\opt{\flow}(\rate))}
=\frac{\coeffalpha\rate+\coeffbeta\rate^{2}}{\coeffgamma+\coeffdelta\rate+\coeffeta\rate^{2}},
\end{equation*}
with $\coeffalpha,\coeffbeta,\coeffdelta,\coeffeta,\ge 0$ and $\coeffgamma\le 0$.
The derivative is given by
\begin{equation*}
\PoA'(\rate)=\frac 1{\SC(\opt{\flow}(\rate))^{2}}\cdot\bracks*{ 
\coeffalpha\coeffgamma+(2\coeffbeta\coeffgamma)\rate+(\coeffbeta\coeffdelta-\coeffalpha\coeffeta)\rate^{2}}
\end{equation*}
and, since $\coeffalpha\coeffgamma\le 0$ and $2\coeffbeta\coeffgamma\le 0$, it can have 
at most one positive zero, and only if $\coeffbeta\coeffdelta-\coeffalpha\coeffeta>0$. 
Hence, either $\PoA'(\rate)$ has a constant sign over  $(\rate_{\run},\rate_{\run+1})$, or it  
changes from negative to positive if the zero lies on $(\rate_{\run},\rate_{\run+1})$, in which case 
we have a local minimum at this zero.

If the optimum has $\activ{\edges}$-breakpoints in $(\rate_{\run},\rate_{\run+1})$, 
we can repeat the argument on each subinterval, noting that $\PoA(\rate)$ is $C^1$
so that the sign $\PoA'(\rate)$ does not change at these $\activ{\edges}$-breakpoints.   
\end{proof}

\cref{pr:UnimodalityAffine} shows that the typical profile of the Price of Anarchy for networks with affine 
costs is similar to the one shown in the example \cref{fi:oscillations-b}. In particular, it implies:
\begin{corollary}
For networks with affine costs the maximum of the Price of Anarchy is attained at some 
$\activ{\edges}$-breakpoint.
\end{corollary}

%----------------------------------------------------------------------
%%% EXAMPLES
%----------------------------------------------------------------------

\section{Examples and counterexamples}
\label{se:examples}

In this section we present a set of examples that illustrate the results of the previous sections.

The first example shows the difference between the set of $\activ{\edges}$-breakpoints and the  
demands at which the set of  paths used at equilibrium changes.

\begin{example}
\label{re:break-points}
Given a selection of equilibrium flows $\eq{\flowprof}(\rate)$, the values $\rate$ at which the set of used paths changes may be different from the $\activ{\edges}$-breakpoints. 
For instance, in the network of \cref{fi:Braess-plus-direct}, 
for all $\rate\geq 2$ and for any choice of $\coeffeta\in[0,1]$ the following is an equilibrium 
\begin{center}
\begin{tabular}{c|c|c|c|c}
path & $\{\source,\vertex_{1},\vertex_{2},\sink\}$ & $\{\source,\vertex_{1},\sink\}$ & $\{\source,\vertex_{2},\sink\}$ & $\{\source,\sink \}$\\
\hline
flow & $1-\coeffeta$ & $\coeffeta$ & $\coeffeta$ & $\rate-1-\coeffeta$
\end{tabular}.
\end{center}
By letting $\coeffeta$  oscillate between $0$ and $1$ arbitrarily often, the set of used paths might
change an arbitrary number of times, whereas the active network is always the set of 
edges in all four paths, each of them with an equilibrium cost equal to $2$.

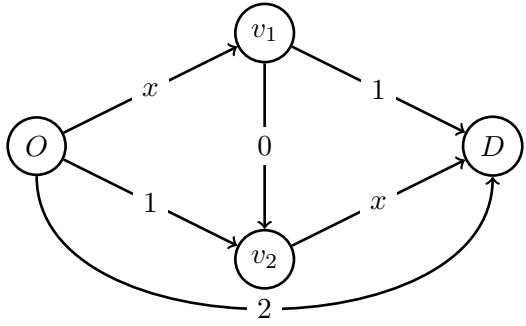
\begin{figure}[ht]

\begin{center}
\begin{tikzpicture}
   \node[shape=circle,draw=black,line width=1pt,minimum size=0.5cm] (v1) at (-3,0)  { $\source$}; 
   \node[shape=circle,draw=black,line width=1pt,minimum size=0.5cm] (v2) at (0,1.5)  {$\vertex_{1}$}; 

   \node[shape=circle,draw=black,line width=1pt,minimum size=0.5cm] (v5) at (0,-1.5)  {$\vertex_{2}$}; 
   \node[shape=circle,draw=black,line width=1pt,minimum size=0.5cm] (v6) at (3,0)  {$\sink$}; 
    
   \draw[line width=1pt,->] (v1) to   node[midway,fill=white] {$\load$} (v2);
   \draw[line width=1pt,->] (v1) to   node[midway,fill=white] {$1$} (v5);
   \draw[line width=1pt,->] (v2) to   node[midway,fill=white] {$0$} (v5);
   
   \draw[line width=1pt,->] (v2) to   node[midway,fill=white] {$1$} (v6);

   \draw[line width=1pt,->] (v5) to   node[midway,fill=white] {$\load$} (v6);
   \draw[line width=1pt,->] (v1) to [bend right=90]  node[midway,fill=white] {$2$} (v6);     
\end{tikzpicture}
\end{center}
\caption{The set of  paths used in equilibrium may change arbitrarily often. 
\label{fi:Braess-plus-direct}
}

\end{figure} 

\end{example}

Our second example shows that in general the set of nonregular demands may be strictly 
larger than the set of $\activ{\edges}$-breakpoints.

\begin{example}
\label{ex:WN}
Consider again \cref{fi:Braess-plus-direct} with the cost $2$ in the lower link replaced by $2+x$,
with $\activ{\edges}$-breakpoints at $\rate=1$ and $\rate=2$.
Note that all the demands $\rate<1$ and $\rate>2$ are regular.
However, for $\rate\in(1,2)$ the active network $\activ{\edges}(\rate)$ is constant and 
comprises all the links, so that there are no $\activ{\edges}$-breakpoints, while the unique equilibrium sends a 
zero flow on the lower link and hence $\rate$ is not regular.
It is worth noting that, although the loss of regularity in the interval $(1,2)$ precludes the use of \cref{pr:EqCost_Differentiable}, 
in this case the costs are affine so that the equilibrium flows are piecewise affine and differentiable 
except at the $\activ{\edges}$-breakpoints $\rate=1$ and $\rate=2$.
\end{example}

The next example deals with the fact that the \ac{PoA} can attain the value $1$ several times.

\begin{example}
\label{ex:nested-Wheatstone}
In  the example of \cref{fi:oscillations}, the \acl{PoA} shows an initial phase in which it is identically equal to $1$, after which it oscillates and eventually goes back to 1 but only asymptotically.
We will use an example taken from \citet[Section~6.1.2]{KliWar:MORfrth} to show that the \ac{PoA} can oscillate and go back to $1$ more than once. 
The idea is to nest several Wheatstone networks and choose constant costs that increase exponentially as we go from the inner to the outer networks.

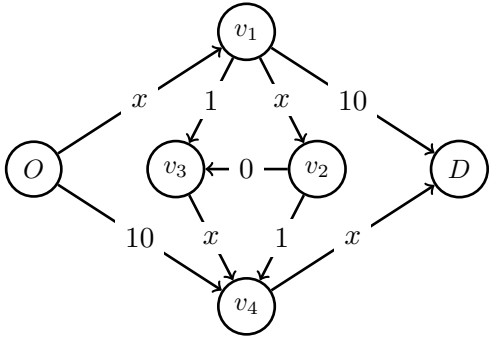
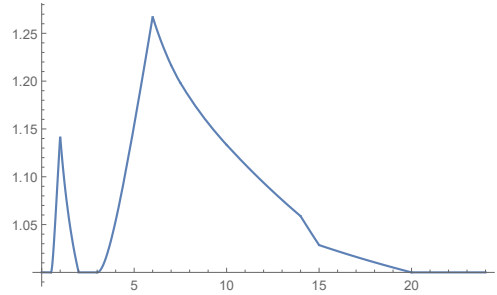
\begin{figure}[h]
\setcounter{subfigure}{0}
\subfigure[Two nested Wheatstone networks]
{
\begin{tikzpicture}[scale=0.7]
    \node[shape=circle,draw=black,line width=1pt] (v1) at (-4,0)  {\small $\source$}; 
   \node[shape=circle,draw=black,line width=1pt] (v2) at (0,2.6)  {\small $\vertex_{1}$}; 
   \node[shape=circle,draw=black,line width=1pt] (v3) at (-1.33,0)  {\small $\vertex_{3}$}; 
   \node[shape=circle,draw=black,line width=1pt] (v4) at (1.33,0)  {\small $\vertex_{2}$}; 
   \node[shape=circle,draw=black,line width=1pt] (v5) at (0,-2.6)  {\small $\vertex_{4}$}; 
   \node[shape=circle,draw=black,line width=1pt] (v6) at (4,0)  {\small $\sink$}; 
   \draw[line width=1pt,->] (v1) to   node[midway,fill=white] {$\load$} (v2);
   \draw[line width=1pt,->] (v1) to   node[midway,fill=white] {$10$} (v5);
   \draw[line width=1pt,->] (v2) to   node[midway,fill=white] {$1$} (v3);
   \draw[line width=1pt,->] (v2) to   node[midway,fill=white] {$\load$} (v4);
   \draw[line width=1pt,->] (v2) to   node[midway,fill=white] {$10$} (v6);
   \draw[line width=1pt,->] (v3) to   node[midway,fill=white] {$\load$} (v5);
   \draw[line width=1pt,->] (v4) to   node[midway,fill=white] {$0$} (v3);
   \draw[line width=1pt,->] (v4) to   node[midway,fill=white] {$1$} (v5);
   \draw[line width=1pt,->] (v5) to   node[midway,fill=white] {$\load$} (v6);
\end{tikzpicture}
\label{fi:2-nested-Wheatstone-extended-graph}
}
\hspace{1cm}
\subfigure[Corresponding \ac{PoA}]
{
\includegraphics[width=0.4\textwidth]{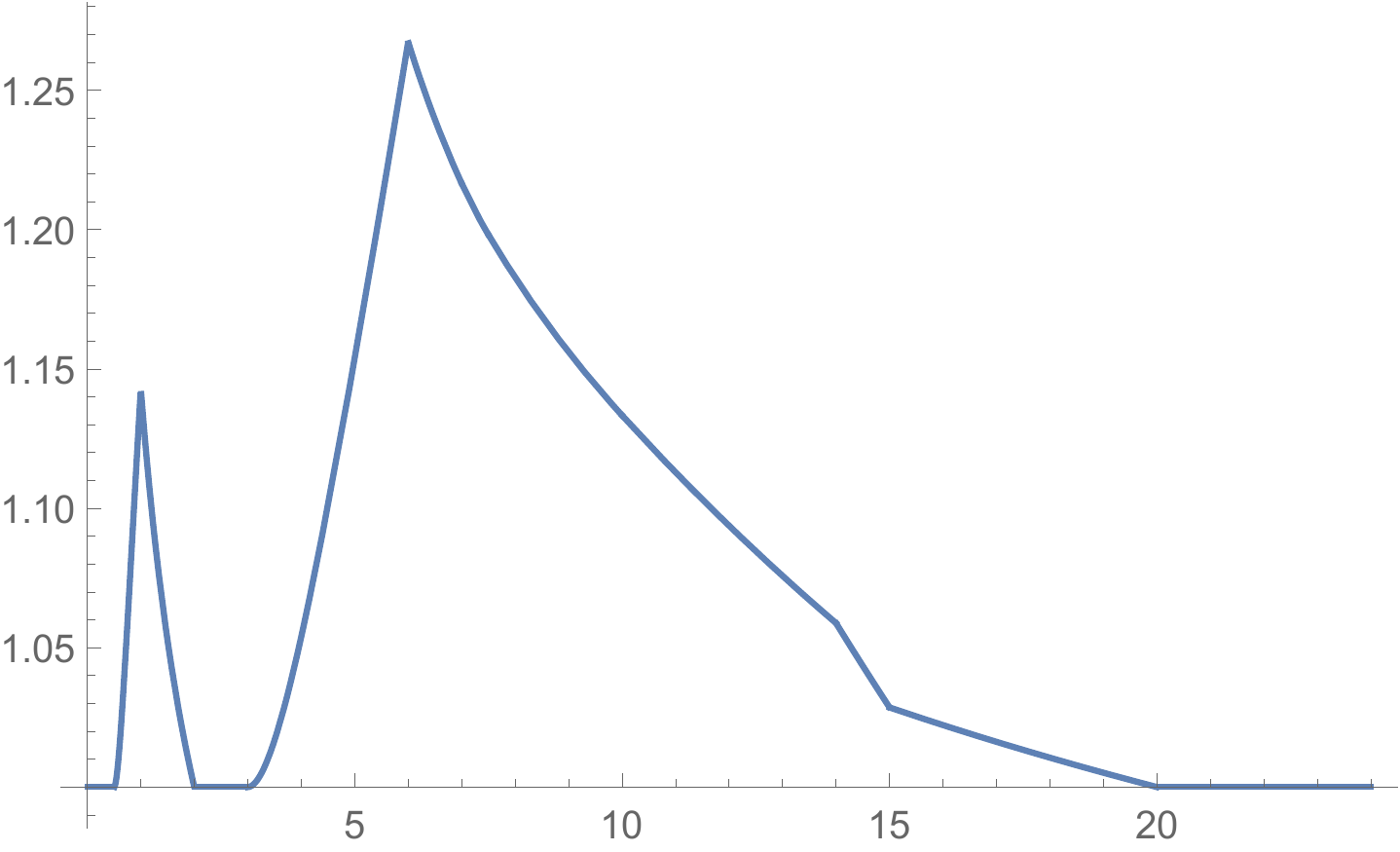}
\label{fi:2-nested-Wheatstone-extended-PoA}
}
\caption{An example where $\PoA$ goes back to $1$ once at intermediate demands}
\label{fi:nested-Wheatstone-extended}

\end{figure}

\cref{fi:2-nested-Wheatstone-extended-graph} shows the version where the network is obtained by nesting with two Wheatstone networks.
\cref{fi:2-nested-Wheatstone-extended-PoA} shows the graph of the corresponding \ac{PoA}.
The \ac{PoA} is equal to $1$ for small demand ($\rate\le 1/2$), then increases and reaches a local maximum, then it decreases back to $1$ and it remains equal to $1$ for the entire interval $[2,3]$ of demand, then reaches its maximum and,  after that, decreases back to $1$, where it remains indefinitely.

Below we list the paths in the network in \cref{fi:2-nested-Wheatstone-extended-graph}:
%\begin{align*}
%\route_{1}=&~\source_5\rightarrow\source_{3}\rightarrow\sink_5,\\
%\route_{2}=&~\source_5\rightarrow\sink_{3}\rightarrow\sink_5,\\
%\route_{3}=&~\source_5\rightarrow\source_{3}\rightarrow\source_{1}\rightarrow\sink_{3}\rightarrow\sink_5,\\
%\route_4=&~\source_5\rightarrow\source_{3}\rightarrow\sink_{1}\rightarrow\sink_{3}\rightarrow\sink_5,\\
%\route_5=&~\source_5\rightarrow\source_{3}\rightarrow\source_{1}\rightarrow\sink_{1}\rightarrow\sink_{3}\rightarrow\sink_5.
%\end{align*}
\begin{align*}
\route_{1}=&~\source\rightarrow\vertex_{1}\rightarrow\sink,\\
\route_{2}=&~\source\rightarrow\vertex_{4}\rightarrow\sink,\\
\route_{3}=&~\source\rightarrow\vertex_{1}\rightarrow\vertex_{2}\rightarrow\vertex_{4}\rightarrow\sink,\\
\route_{4}=&~\source\rightarrow\vertex_{1}\rightarrow\vertex_{3}\rightarrow\vertex_{4}\rightarrow\sink,\\
\route_{5}=&~\source\rightarrow\vertex_{1}\rightarrow\vertex_{2}\rightarrow\vertex_{3}\rightarrow\vertex_{4}\rightarrow\sink.
\end{align*}
The equilibrium flow for $\rate\in [0,+\infty)$ is given explicitly in the following table:

\medskip

\begin{center}

\def\arraystretch{2}
\begin{tabular}{c|c|c|c|c|c|c}
Interval & Cost $\eqcost(\rate)$ & $\route_{1}$ & $\route_{2}$ & $\route_{3}$ & $\route_4$ & $\route_5$ \\
\hline
$\rate\in[0,1)$ & $4\rate$ & 0 & 0 & 0 & 0 & $\rate$  \\
\hline
$\rate\in[1,2)$ & $2+2\rate$ & 0 & 0 & $\rate-1$ & $\rate-1$ & $2-\rate$  \\
\hline
$\rate\in[2,6)$& $\displaystyle 1 +\frac 52 \rate$ & 0 & 0 & $\displaystyle\frac{\rate}{2}$ & $\displaystyle\frac{\rate}{2}$ & 0 \Bstrut \\
\hline
$\rate\in[6,14)$ &$\displaystyle \frac{29}{2}+\frac{\rate}{4}$& $\displaystyle\frac 34 \rate -\frac 92$ & $\displaystyle\frac 34 \rate -\frac 92$ & $\displaystyle \frac 92-\frac{\rate}{4}$ & $\displaystyle \frac 92-\frac{\rate}{4}$ & 0 \Bstrut \\
\hline
$\rate\in[14,15)$ &$18$& $\rate-8$ & $\rate-8$ & $15-\rate$ & $15-\rate$ & $\rate-14$  \\
\hline
$\rate\in[15,20)$ &$\displaystyle 12+\frac 25 \rate$& $\displaystyle \frac 35 \rate-2$ & $\displaystyle \frac 35 \rate-2$ & 0 & 0 & $\displaystyle 4 -\frac{\rate}{5}$ \Bstrut  \\
\hline
$\rate\ge 20$ &$\displaystyle 10+\frac{\rate}{2}$& $\displaystyle\frac{\rate}{2}$ & $\displaystyle\frac{\rate}{2}$ & 0 & 0 & 0  \\

\end{tabular}

\end{center}

\bigskip

As a consequence, the \ac{PoA} is
\begin{equation}
\label{eq:PoA-nested-Braess}
\PoA(\rate)=
\begin{cases}
1 &\text{if }\rate\in[0,\frac 12)\\
\dfrac{8\rate^{2}}{-1+4\rate+4\rate^{2}} &\text{if }\rate\in[\frac 12,1)\\[10pt]

\dfrac{4+4\rate}{2+5\rate} &\text{if }\rate\in[1,2)\\[10pt]

1 &\text{if }\rate\in[2,3)\\[10pt]

\dfrac{4\rate+10\rate^{2}}{-81+58\rate+\rate^2} &\text{if }\rate\in[3,6)\\[10pt]

\dfrac{58\rate+\rate^{2}}{-81+58\rate+\rate^{2}} &\text{if }\rate\in[6,7)\\[10pt]

\dfrac{58\rate+\rate^2}{-130+72\rate} &\text{if }\rate\in[7,\frac{15}{2})\\[10pt]

\dfrac{290\rate+5\rate^2}{-200+240\rate+8\rate^2} &\text{if }\rate\in[\frac{15}{2},10)\\[10pt]

\dfrac{58+\rate}{40+2\rate} &\text{if }\rate\in[10,14)\\[10pt]

\dfrac{36}{20+\rate} &\text{if }\rate\in[14,15)\\[10pt]

\dfrac{120+4\rate}{100+5\rate} &\text{if }\rate\in[15,20)\\[10pt]

1 &\text{if }\rate\ge 20.
\end{cases}
\end{equation}

The case of three nested Wheatstone networks can be treated similarly (see \cref{fi:3-nested-Wheatstone-extended}). Here the \ac{PoA} is $1$ in the intervals $[0,\frac 12]$, $[2,3]$, $[20,30]$ and $[200,+\infty)$.

\begin{figure}[h]
\setcounter{subfigure}{0}
\subfigure[Three nested Wheatstone networks]
{
\begin{tikzpicture}[scale=0.7]
\node[shape=circle,draw=black,line width=1pt] (v0) at (-6,0)  {\small $\source$};
    \node[shape=circle,draw=black,line width=1pt] (v1) at (0,4)  {\small $\vertex_{1}$}; 
   \node[shape=circle,draw=black,line width=1pt] (v2) at (2.6,0)  {\small $\vertex_{2}$}; 
   \node[shape=circle,draw=black,line width=1pt] (v3) at (0,1.33)  {\small $\vertex_{4}$}; 
   \node[shape=circle,draw=black,line width=1pt] (v4) at (0,-1.33)  {\small $\vertex_{3}$}; 
   \node[shape=circle,draw=black,line width=1pt] (v5) at (-2.6,0)  {\small $\vertex_{5}$}; 
   \node[shape=circle,draw=black,line width=1pt] (v6) at (0,-4)  {\small $\vertex_{6}$}; 
   \node[shape=circle,draw=black,line width=1pt] (v7) at (6,0)  {\small $\sink$}; 
   \draw[line width=1pt,->] (v1) to   node[midway,fill=white] {$\load$} (v2);
   \draw[line width=1pt,->] (v1) to   node[midway,fill=white] {$10$} (v5);
   \draw[line width=1pt,->] (v2) to   node[midway,fill=white] {$1$} (v3);
   \draw[line width=1pt,->] (v2) to   node[midway,fill=white] {$\load$} (v4);
   \draw[line width=1pt,->] (v2) to   node[midway,fill=white] {$10$} (v6);
   \draw[line width=1pt,->] (v3) to   node[midway,fill=white] {$\load$} (v5);
   \draw[line width=1pt,->] (v4) to   node[midway,fill=white] {$0$} (v3);
   \draw[line width=1pt,->] (v4) to   node[midway,fill=white] {$1$} (v5);
   \draw[line width=1pt,->] (v5) to   node[midway,fill=white] {$\load$} (v6);
   \draw[line width=1pt,->] (v0) to   node[midway,fill=white] {$\load$} (v1);
    \draw[line width=1pt,->] (v0) to   node[midway,fill=white] {$100$} (v6);
     \draw[line width=1pt,->] (v1) to   node[midway,fill=white] {$100$} (v7);
    \draw[line width=1pt,->] (v6) to   node[midway,fill=white] {$\load$} (v7);
   
\end{tikzpicture}
\label{fi:3-nested-Wheatstone-extended-graph}
}
\hspace{1cm}
\subfigure[Corresponding \ac{PoA}]
{
\includegraphics[width=0.6\textwidth]{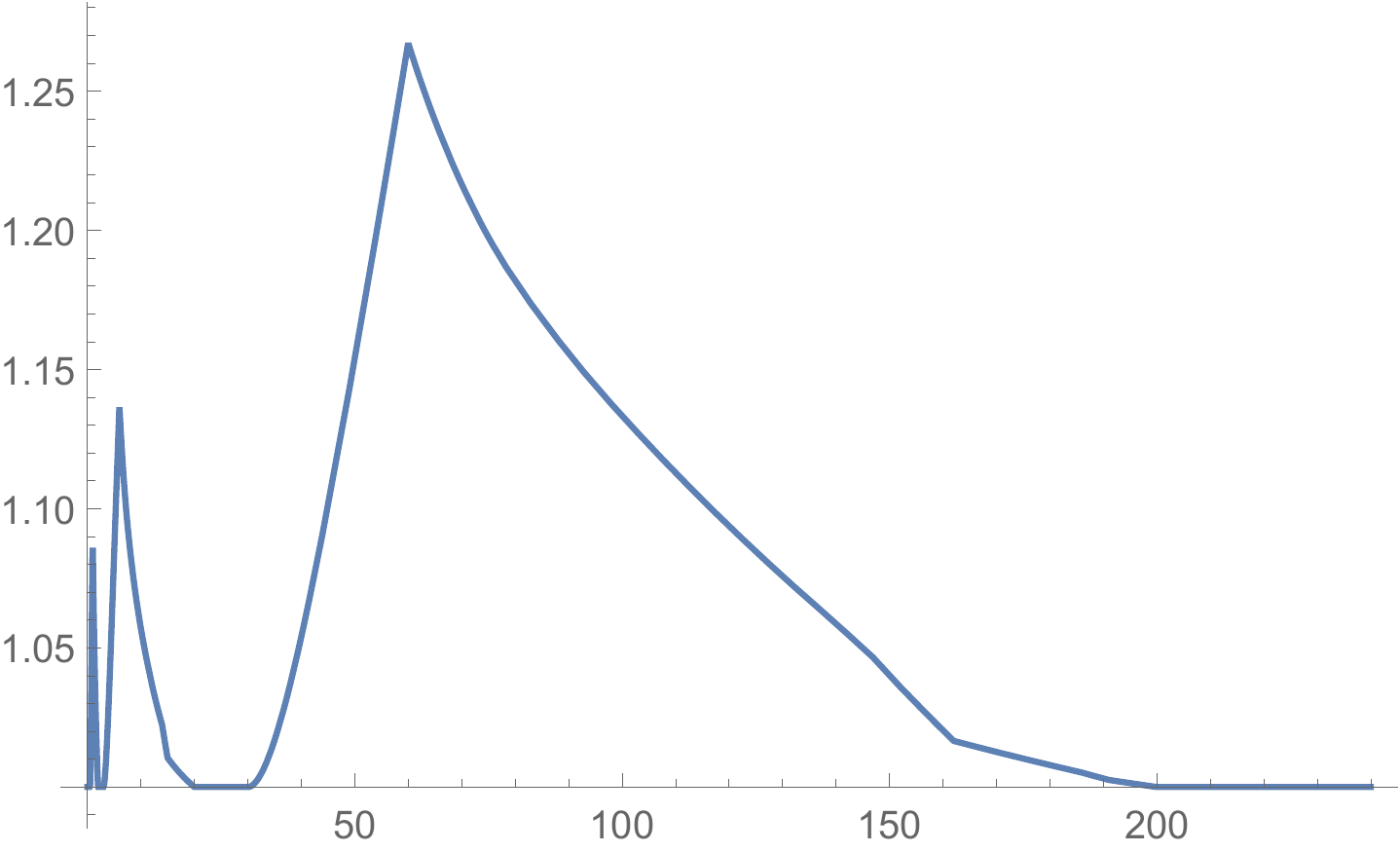}
\label{fi:3-nested-Wheatstone-extended-PoA}
}
\caption{An example where $\PoA$ goes back to $1$ twice at intermediate demands}
\label{fi:3-nested-Wheatstone-extended}
\end{figure}
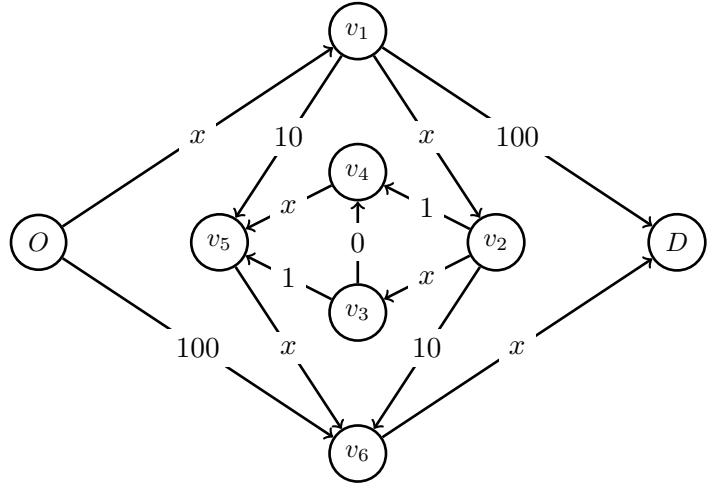
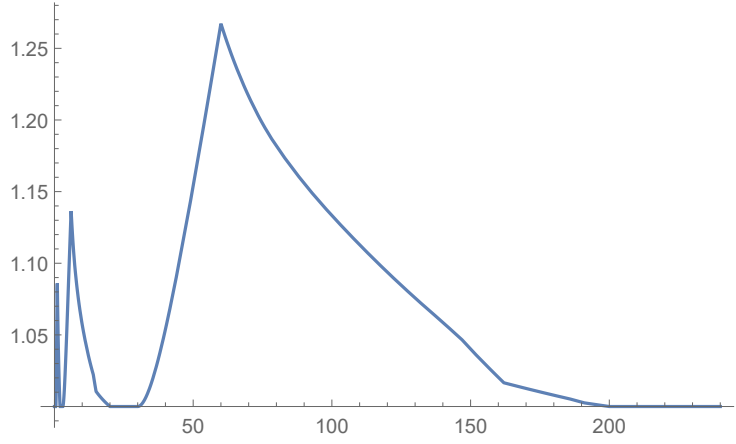
\end{example}

\begin{example}
This example shows that the result in  \cref{pr:UnimodalityAffine} fails for polynomials, even for the simplest network topology.
Indeed, consider the two-link parallel network with cost functions
\begin{equation}
\label{eq:parallel-costs}
\begin{split}
\cost_{1}(\load_{1})&=\load_{1}, \\
\cost_{2}(\load_{2})&=1+\load_{2}^2.
\end{split}
\end{equation}
The equilibrium and optimum flows can be computed explicitly as:

\medskip

\begin{center}
\begin{tabular}{c|c|c|c|c}
demand & $\eq\load_{1}(\rate)$ & $\eq\load_{2}(\rate)$ & $\opt\load_{1}(\rate)$ & $\opt\load_{2}(\rate)$ \Bstrutsmall\\
\hline
$\rate\in[0,1/2]$ & $\rate$ & 0 & $\rate$ & 0\\
\hline
$\rate\in[1/2,1]$ & $\rate$ & 0 &  $\rate-\dfrac{-1+\sqrt{6\rate-2}}{3}$ & $\dfrac{-1+\sqrt{6\rate-2}}{3}$ \Tstrut \Bstrut\\
\hline
$\rate\in[1,+\infty)$ & $\rate-\dfrac{-1+\sqrt{4\rate-3}}{2}$ & $\dfrac{-1+\sqrt{4\rate-3}}{2} $&  $\rate-\dfrac{-1+\sqrt{6\rate-2}}{3}$ & $\dfrac{-1+\sqrt{6\rate-2}}{3}$ \Tstrut
\end{tabular}
\end{center}

\medskip

\noindent Note that the sole $\activ{\edges}$-breakpoint is at $\rate=1$. 
Moreover we have $\eq\load_{1}(3)=\opt\load_{1}(3)=2$ and $\eq\load_{2}(3)=\opt\load_{2}(3)=1$ so that
the equilibrium and optimal flows coincide, and $\PoA(3)=1$. 
This implies that in the interval $(3,+\infty)$ the \ac{PoA} has a local maximum (see \cref{fi:parallel-PoA}).

\begin{figure}[h]
\begin{center}
\includegraphics[width=0.5\textwidth]{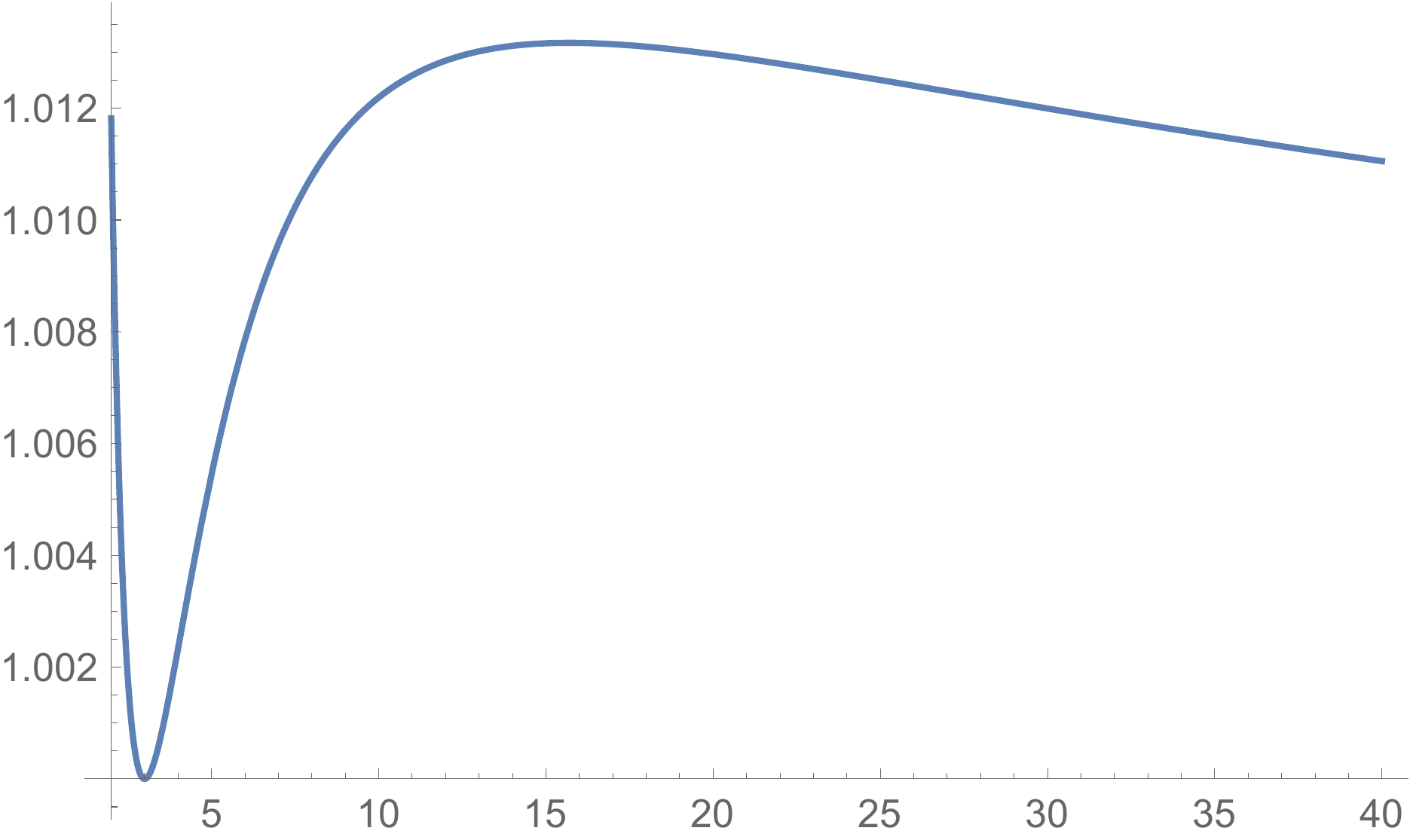}
\caption{Plot of the \ac{PoA} on the interval $[2,40]$ for the parallel network with costs as in \cref{eq:parallel-costs}.}
\label{fi:parallel-PoA}
\end{center}
\end{figure}
\end{example}

When the cost functions are less regular, the set of paths used at equilibrium can have a recurring behavior,
and an active network that is abandoned at some point can be reactivated at larger demands.
In particular we cannot ensure that the number of $\activ{\edges}$-breakpoints is finite. 

\begin{proposition}
\label{pr:RepeatingActiveRegimes}
There exist networks and nondecreasing cost functions $\cost_{\edge}(\argdot)$ such that a given active network $\activ{\edges}(\rate)=\edges_{0}$ can repeat itself over disjoint demand intervals defined by $\activ{\edges}$-breakpoints.
\end{proposition}

\begin{proof}
Consider the network in \cref{fi:non-affine-Wheatstone} with the cost $\cost(\load)$ defined in $[0,\coeffzeta]\cup[\coeffzeta+\coeffeta,+\infty]$ as follows
\begin{equation*}
\cost(\load)=\begin{cases}
\coeffa &\text{ if }\load\le\coeffzeta\\
\coeffb &\text{ if }\load\ge\coeffzeta+\coeffeta,
\end{cases}
\end{equation*}
with $0<\coeffeta<\coeffa<\coeffzeta<\coeffzeta+\coeffeta<\coeffb$, and in the interval $[\coeffzeta,\coeffzeta+\coeffeta]$ we interpolate in any way that makes $\cost(\load)$ continuous and nondecreasing in the whole $[0,+\infty]$. 

\begin{figure}

\begin{center}
\begin{tikzpicture}
   \node[shape=circle,draw=black,line width=1pt,minimum size=0.5cm] (v1) at (-3,0)  { $\source$}; 
   \node[shape=circle,draw=black,line width=1pt,minimum size=0.5cm] (v2) at (0,2)  { $\vertex_{1}$}; 

   \node[shape=circle,draw=black,line width=1pt,minimum size=0.5cm] (v5) at (0,-2)  { $\vertex_{2}$}; 
   \node[shape=circle,draw=black,line width=1pt,minimum size=0.5cm] (v6) at (3,0)  { $\sink$}; 
   
   \draw[line width=1pt,->] (v1) to   node[above=.05cm, midway] {$\load$} (v2);
   \draw[line width=1pt,->] (v1) to   node[below=.05cm, midway] {$\cost(\load)$} (v5);
   \draw[line width=1pt,->] (v2) to   node[left=.05cm, midway] {$0$} (v5);
   
   \draw[line width=1pt,->] (v2) to   node[above=.05cm,midway] {$\cost(\load)$} (v6);

   \draw[line width=1pt,->] (v5) to   node[below=.05cm,midway] {$\load$} (v6);

\end{tikzpicture}
\end{center}

\caption{Wheatsone network with nonaffine costs.
\label{fi:non-affine-Wheatstone} }

\end{figure}
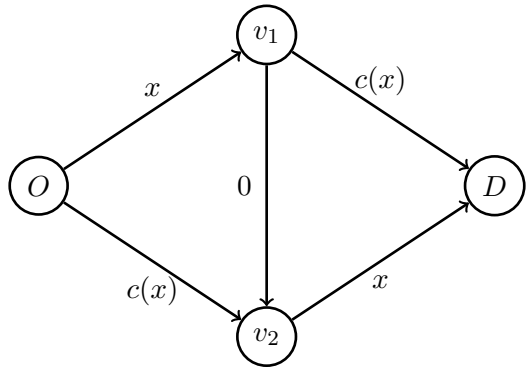

Then we have the following regimes:
\begin{enumerate}[(a)]
\item 
\label{it:regime-0-a}
when $\rate\in [0,\coeffa]$, the equilibrium flow uses only the path $\{\source,\vertex_{1},\vertex_{2},\sink\}$, the load on the two edges $\{\source,\vertex_{2}\}$, and $\{\vertex_{1},\sink\}$ is zero and the equilibrium cost is $\eqcost(\rate)=2\rate$;

\item 
\label{it:regime-a-2a}
when $\rate\in(\coeffa,2\coeffa)$, the equilibrium flow uses all the three paths with the following distribution
\begin{center}
\begin{tabular}{c|c|c|c}
path & $\{\source,\vertex_{1},\vertex_{2},\sink\}$ & $\{\source,\vertex_{1},\sink\}$ & $\{\source,\vertex_{2},\sink\}$\\
\hline
flow & $2\coeffa-\rate$ & $\rate-\coeffa$ & $\rate-\coeffa$ 
\end{tabular}
\end{center}
the load on the two edges  $\{\source,\vertex_{2}\}$, $\{\vertex_{1},\sink\}$ is $\rate-\coeffa<\coeffzeta$, and the equilibrium cost is $\eqcost(\rate)=2\coeffa$;

\item 
\label{it:regime-2a-2g}
when $\rate\in [2\coeffa,2\coeffzeta]$, the equilibrium flow splits equally between the two paths $\{\source,\vertex_{1},\sink\}$, $\{\source,\vertex_{2},\sink\}$, the load on the two edges  $\{\source,\vertex_{2}\}$, $\{\vertex_{1},\sink\}$ is $\rate/2\le\coeffzeta$, and the equilibrium cost is $\eqcost(\rate)=\coeffa+\rate/2$;

\item
\label{it:regime-2g-b+g+e}
the regime on the demand interval  $[2\coeffzeta,\coeffb+\coeffzeta+\coeffeta]$ is complicated to describe, but this is not relevant for our purpose;

\item 
\label{it:regime-b+g+e-2b}
when $\rate\in[\coeffb+\coeffzeta+\coeffeta,2\coeffb)$, the equilibrium  uses all  three paths with the following distribution
\begin{center}
\begin{tabular}{c|c|c|c}
path & $\{\source,\vertex_{1},\vertex_{2},\sink\}$ & $\{\source,\vertex_{1},\sink\}$ & $\{\source,\vertex_{2},\sink\}$\\
\hline
flow & $2\coeffb-\rate$ & $\rate-\coeffb$ & $\rate-\coeffb$ 
\end{tabular}
\end{center}
the load on the two edges  $\{\source,\vertex_{2}\}$, $\{\vertex_{1},\sink\}$ is $\rate-\coeffb\ge\coeffzeta+\coeffeta$, and the equilibrium cost is $\eqcost(\rate)=2\coeffb$;

\item 
\label{it:regime2b-infty}
when $\rate\ge 2\coeffb$, the equilibrium flow splits equally between the two paths $\{\source,\vertex_{1},\sink\}$, $\{\source,\vertex_{2},\sink\}$, and the load on the two edges  $\{\source,\vertex_{2}\}$, $\{\vertex_{1},\sink\}$ is $\rate/2>\coeffzeta+\coeffeta$.
\end{enumerate}

This proves our claim, as in the intervals sub \ref{it:regime-2a-2g} and \ref{it:regime2b-infty}  the equilibrium uses only the two paths $\{\source,\vertex_{1},\sink\}$ and $\{\source,\vertex_{2},\sink\}$, while in the interval sub \ref{it:regime-a-2a} and \ref{it:regime-b+g+e-2b} it uses all three paths.
\end{proof}

\begin{remark}
\label{re:repeat}
Note that in the same way one can construct examples of networks with an infinite number of $\activ{\edges}$-breakpoints for the equilibrium. 
Furthermore, one can make the increasing sequence of such $\activ{\edges}$-breakpoints to be convergent. Indeed, one could chose infinite sequences $(\coeffeta_{\irun})_{\irun\in\N}$, $(\coeffa_{\irun})_{\irun\in\N}$, $(\coeffzeta_{\irun})_{\irun\in\N}$, $(\coeffb_{\irun})_{\irun\in\N}$ with 
\begin{equation}
\label{eq:parameter-constraints}
0<\coeffeta_{\irun}<\coeffa_{\irun}<\coeffzeta_{\irun}<\coeffzeta_{\irun}+\coeffeta_{\irun}<\coeffb_{\irun} 
<\coeffa_{\irun+1}\quad\text{for every }\irun\in\N.
\end{equation}
Then we set the cost function $\cost(\load)$ to be
\begin{equation*}
\cost(\load)=
\begin{cases}
\coeffa_{1} &\text{ if }\load\le\coeffzeta_{1}\\
\coeffb_{\irun} &\text{ if }\coeffzeta_{\irun}+\coeffeta_{\irun}\le\load\le\dfrac{\coeffa_{\irun+1}-\coeffb_{\irun}}2\\
\coeffa_{\irun+1} &\text{ if }\coeffa_{\irun+1}\le\load\le\coeffzeta_{\irun+1},
\end{cases}
\end{equation*}
and in the intervals $[\coeffzeta_{\irun},\coeffzeta_{\irun}+\coeffeta]$, $[(\coeffa_{\irun+1}-\coeffb_{\irun})/2,\coeffa_{\irun+1}]$ we interpolate in any way that makes $\cost(\load)$ continuous and nondecreasing in the whole $[0,+\infty]$. With this choice of $\cost(\load)$, the  situation is similar to the one of \cref{pr:RepeatingActiveRegimes}, repeated infinitely many times. 
Furthermore, since the $\coeffeta_{\irun}$ can be as small as we want, we can choose the sequences $(\coeffa_{\irun})_{\irun\in\N}$, $(\coeffzeta_{\irun})_{\irun\in\N}$, $(\coeffb_{\irun})_{\irun\in\N}$ to be convergent to the same limit.
\end{remark}

\begin{remark}
\label{re:piecewise-affine}
Repetition of an active network can happen for cost functions that are smooth or piecewise affine, as in \citet{KliWar:MORfrth}. 
We know that this cannot happen for affine cost functions, but, so far we have not been able to characterize the class of cost functions for which repetitions of active networks are impossible.
\end{remark}

%----------------------------------------------------------------------
%%% ACKNOWLEDGMENTS
%----------------------------------------------------------------------

\subsection*{Acknowledgments}
Marco Scarsini and Valerio Dose are members of INdAM-GNAMPA.
Roberto Cominetti gratefully acknowledges the support of Luiss University during a visit in which this research was initiated, as well as the support of
Proyecto Anillo ANID/PIA/ACT192094.
This research project received partial support from the COST action GAMENET, the INdAM-GNAMPA Project 2020 \emph{Random Walks on Random Games}, and  the Italian MIUR PRIN 2017 Project ALGADIMAR \emph{Algorithms, Games, and Digital Markets}.

The authors thank the reviewers for their careful reading of the paper, for their useful suggestions,  and for bringing to their attention the  article by \citet{KliWar:MORfrth}.

%----------------------------------------------------------------------
%%% SYMBOLS
%----------------------------------------------------------------------

\section{List of symbols}
\label{app:symbols}

\begin{longtable}{p{.13\textwidth} p{.82\textwidth}}

$\coeffa_{\edge}$ & coefficient of the monomial in an affine cost function \\

$\matrA$ & $\epincid^{\top}\diaga\epincid$, defined in \cref{eq:A-beta} \\

$\coeffb_{\edge}$ & constant in an affine cost function \\

$\coeffbprof$ & $(\coeffb_{1},\dots,\coeffb_{\nedges})^{\top}$ \\

$\cost_{\edge}$ & cost of edge $\edge$ \\

$\costprof$ & $(\cost_{1},\dots,\cost_{\nedges})^{\top}$ \\

$\opt{\cost}_{\edge}(\load_{\edge})$ & $\opt{\Cost}_{\edge}'(\load_{\edge})$, defined in \cref{eq:marginal-cost} \\

$\Cost_{\edge}$ & primitive of $\cost_{\edge}$, defined in \cref{eq:Cost} \\

$\opt{\Cost}_{\edge}(\load_{\edge})$ & $\load_{\edge}\, \cost_{\edge}(\load_{\edge})$ \\

$\cost_{\route}$ & cost of path $\route$, defined in \cref{eq:cost-path} \\

$\sumbvec$ & $\epincid^{\top}\coeffbprof$, defined in \cref{eq:A-beta} \\

$\sink$ & destination of the network \\

$\dual_{\rate}$ & dual problem \\

$\edge$ & edge \\

$\edges$ & set of edges \\

$\activ{\edges}(\rate)$ & active network at demand $\rate$, defined in \cref{de:activenetwork} \\

$\activ{\edges}_{0}$ & $\activ{\edges}(\rate_{0})$ \\

$\edges\parens*{\vertex,\vertexalt}$ & set of all edges $\edge\in\edges$ that go from $\vertex$ to $\vertexalt$ \\

$\activ{\edges}_{0}\parens*{\vertex,\vertexalt}$ &  set of all edges $\edge\in\activ{\edges}_{0}$ that go from $\vertex$ to $\vertexalt$ \\

$\flowprof$ & flow profile \\

$\eq{\flowprof}$ & equilibrium flow profile \\

$\opt{\flowprof}$ & optimum flow profile \\

$\flow_{\route}$ & flow of path $\route$ \\

$\flows_{\rate}$ & set of flows of total demand $\rate$, defined in \cref{eq:flows}\\

$\graph$ & graph \\

$\nedges$ & $\card{\edges}$ \\

$\npaths$ & $\card{\routes}$ \\

$\outneigh(\vertex)$ & out-edges of vertex $\vertex$ \\

$\inneigh(\vertex)$ & in-edges of vertex $\vertex$ \\

$\source$ & origin of the network \\

$\route$ & path \\

$\routes$ & set of paths \\

$\prim_{\rate}$ & primal problem \\

$\PoA$ & price of anarchy \\

$\solset$ & solution set \\

$\SC$ & social cost, defined in \cref{eq:SC} \\

$\incrcostedge_{\edge}$ & increment of $\eqcostedge_{\edge}$, defined in the proof of \cref{pr:EqCost_Differentiable} \\

$\eqcostvertex_{\vertex}$ & equilibrium cost of shortest path to $\vertex$, defined in \cref{eq:T-min} \\

$\incrload_{\edge}$ & increment of $\load_{\edge}$, defined in the proof of \cref{pr:EqCost_Differentiable} \\

$\vertices$ & set of vertices \\

$\vertices_{0}$ & set of vertices in the active network \\

$\vinf_{\rate}(\zvar)$ & $\inf_{\flow}\perturb_{\rate}(\flowprof,z)$, defined in \cref{eq:v-m} \\

$\valueW(\rate)$ & solution of the equilibrium minimization problem, defined in \cref{eq:Wardrop-variational} \\

$\opt{\valueW}(\rate)$  & solution of the optimum minimization problem, defined in \cref{eq:opt-min} \\

$\loadprof$ & load profile \\

$\eq{\loadprof}$ & equilibrium load profile \\

$\load_{\edge}$ & load of edge $\edge$ \\

$\loads_{\rate}$ & set of loads of total demand $\rate$ \\

$\epincid$ &  edge-path incidence matrix \\

$\coeffalpha$ & element of $\reals_{+}$, first used in \cref{pr:affine-flows} \\

$\coeffalpha_{\run}$ & element of $\reals_{+}$, first used in \cref{pr:SocialCostAffine} \\

$\coeffbeta$ & element of $\reals_{+}$, first used in \cref{pr:affine-flows} \\

$\coeffbeta_{\run}$ & element of $\reals_{+}$, first used in \cref{pr:SocialCostAffine} \\

$\coeffgamma$ & element of $\reals_{-}$, first used in \cref{pr:UnimodalityAffine} \\ 

$\coeffgamma_{\run}$ & element of $\reals_{-}$, first used in \cref{pr:SocialCostAffine} \\

$\diaga$ & $\diag[(\coeffa_{\edge})_{\edge\in\edges}]$ \\

$\coeffdelta$ & element of $\reals_{+}$, first used in \cref{pr:UnimodalityAffine} \\

$\incrcostvertex_{\vertex}$ & increment of $\eqcostvertex_\vertex$, defined in the proof of \cref{pr:EqCost_Differentiable} \\

$\Delta_{\vertex,\vertexalt}$ & constant travel time from $\vertex$ to $\vertexalt$\\

$\coeffzeta$ & element of $\reals_{+}$, first used in the proof of 
\cref{pr:RepeatingActiveRegimes} \\

$\coeffeta$ & element of $\reals_{+}$, first used in \cref{pr:UnimodalityAffine} \\

$\eqcost$ & equilibrium cost, defined in \cref{eq:Wardrop} \\

$\rate$ & demand \\

$\rate_{0}$ & $\activ{\edges}$-breakpoint for the equilibrium, defined in \cref{de:activenetwork} \\

$\opt{\rate}_{0}$ & $\activ{\edges}$-breakpoint for the optimum \\

$\eqcostedge_{\edge}$ & equilibrium cost of edge $\edge$ \\

$\perturb_{\rate}$ & perturbation function, defined in \cref{eq:perturb} \\

\end{longtable}

\bibliographystyle{apalike}
\bibliography{biblio-games}

\begin{thebibliography}{}

\bibitem[Aliprantis and Border, 2006]{AliBor:Springer2006}
Aliprantis, C.~D. and Border, K.~C. (2006).
\newblock {\em Infinite Dimensional Analysis}.
\newblock Springer, Berlin, third edition.

\bibitem[Beckmann et~al., 1956]{BecMcGWin:Yale1956}
Beckmann, M.~J., McGuire, C., and Winsten, C.~B. (1956).
\newblock {\em Studies in the Economics of Transportation}.
\newblock Yale University Press, New Haven, CT.

\bibitem[Colini-Baldeschi et~al., 2020]{ColComMerSca:OR2020}
Colini-Baldeschi, R., Cominetti, R., Mertikopolous, P., and Scarsini, M.
  (2020).
\newblock When is selfish routing bad? {T}he price of anarchy in light and
  heavy traffic.
\newblock {\em Oper. Res.}, 68(2):411--434.

\bibitem[Colini-Baldeschi et~al., 2019]{ColComSca:TOCS2019}
Colini-Baldeschi, R., Cominetti, R., and Scarsini, M. (2019).
\newblock Price of anarchy for highly congested routing games in parallel
  networks.
\newblock {\em Theory Comput. Syst.}, 63(1):90--113.

\bibitem[Colini-Baldeschi et~al., 2018]{ColKliSca:ICALP2018}
Colini-Baldeschi, R., Klimm, M., and Scarsini, M. (2018).
\newblock Demand-independent optimal tolls.
\newblock In {\em 45th {I}nternational {C}olloquium on {A}utomata, {L}anguages,
  and {P}rogramming}. Schloss Dagstuhl. Leibniz-Zent. Inform., Wadern.

\bibitem[Cominetti et~al., 2019]{ComScaSchSti:EC2019}
Cominetti, R., Scarsini, M., Schr\"oder, M., and Stier-Moses, N. (2019).
\newblock Price of anarchy in stochastic atomic congestion games with affine
  costs.
\newblock In {\em Proceedings of The Twentieth ACM Conference on Economics and
  Computation (ACM EC '19)}.

\bibitem[Correa et~al., 2004]{CorSchSti:MOR2004}
Correa, J.~R., Schulz, A.~S., and Stier-Moses, N.~E. (2004).
\newblock Selfish routing in capacitated networks.
\newblock {\em Math. Oper. Res.}, 29(4):961--976.

\bibitem[Correa et~al., 2007]{CorSchSti:OR2007}
Correa, J.~R., Schulz, A.~S., and Stier-Moses, N.~E. (2007).
\newblock Fast, fair, and efficient flows in networks.
\newblock {\em Oper. Res.}, 55(2):215--225.

\bibitem[Correa et~al., 2008]{CorSchSti:GEB2008}
Correa, J.~R., Schulz, A.~S., and Stier-Moses, N.~E. (2008).
\newblock A geometric approach to the price of anarchy in nonatomic congestion
  games.
\newblock {\em Games Econom. Behav.}, 64(2):457--469.

\bibitem[Correa and Stier-Moses, 2011]{CorSti:EORMS2011}
Correa, J.~R. and Stier-Moses, N.~E. (2011).
\newblock Wardrop equilibria.
\newblock In Cochran, J.~J., editor, {\em Encyclopedia of Operations Research
  and Management Science}. Wiley.

\bibitem[Dafermos and Nagurney, 1984]{DafNag:MP1984}
Dafermos, S. and Nagurney, A. (1984).
\newblock Sensitivity analysis for the asymmetric network equilibrium problem.
\newblock {\em Math. Programming}, 28(2):174--184.

\bibitem[Dafermos and Sparrow, 1969]{DafSpa:JRNBSB1969}
Dafermos, S.~C. and Sparrow, F.~T. (1969).
\newblock The traffic assignment problem for a general network.
\newblock {\em J. Res. Nat. Bur. Standards Sect. B}, 73B:91--118.

\bibitem[Dumrauf and Gairing, 2006]{DumGai:INE2006}
Dumrauf, D. and Gairing, M. (2006).
\newblock Price of anarchy for polynomial {W}ardrop games.
\newblock In {\em WINE '06: Proceedings of the 2nd Conference on Web and
  Internet Economics}, pages 319--330. Springer Berlin Heidelberg, Berlin,
  Heidelberg.

\bibitem[Englert et~al., 2010]{EngFraOlb:TCS2010}
Englert, M., Franke, T., and Olbrich, L. (2010).
\newblock Sensitivity of {W}ardrop equilibria.
\newblock {\em Theory Comput. Syst.}, 47(1):3--14.

\bibitem[Fisk, 1979]{Fis:TRB1979}
Fisk, C. (1979).
\newblock More paradoxes in the equilibrium assignment problem.
\newblock {\em Transportation Res. Part B}, 13(4):305 -- 309.

\bibitem[Florian and Hearn, 1995]{FloHea:HORMS1995}
Florian, M. and Hearn, D. (1995).
\newblock Network equilibrium models and algorithms.
\newblock In et~al, M.~H., editor, {\em Handbooks in Operation Research and
  Management Science}, volume~8. North Holland, Amsterdam.

\bibitem[Fukushima, 1984]{Fuk:TRB1984}
Fukushima, M. (1984).
\newblock On the dual approach to the traffic assignment problem.
\newblock {\em Transportation Res. Part B}, 18(3):235--245.

\bibitem[Gemici et~al., 2019]{GemKouMonPapPil:ISTACS2019}
Gemici, K., Koutsoupias, E., Monnot, B., Papadimitriou, C.~H., and Piliouras,
  G. (2019).
\newblock Wealth inequality and the price of anarchy.
\newblock In {\em 36th {I}nternational {S}ymposium on {T}heoretical {A}spects
  of {C}omputer {S}cience}. Schloss Dagstuhl. Leibniz-Zent. Inform., Wadern.

\bibitem[Hall, 1978]{Hal:TS1978}
Hall, M.~A. (1978).
\newblock Properties of the equilibrium state in transportation networks.
\newblock {\em Transportation Sci.}, 12(3):208--216.

\bibitem[Josefsson and Patriksson, 2007]{JosPat:TRB2007}
Josefsson, M. and Patriksson, M. (2007).
\newblock Sensitivity analysis of separable traffic equilibrium equilibria with
  application to bilevel optimization in network design.
\newblock {\em Transportation Res. Part B}, 41(1):4 -- 31.

\bibitem[Klimm and Warode, 2019]{KliWar:SODA2019}
Klimm, M. and Warode, P. (2019).
\newblock Computing all {W}ardrop equilibria parametrized by the flow demand.
\newblock In {\em Proceedings of the {T}hirtieth {A}nnual {ACM}-{SIAM}
  {S}ymposium on {D}iscrete {A}lgorithms}, pages 917--934. SIAM, Philadelphia,
  PA.

\bibitem[Klimm and Warode, 2021]{KliWar:MORfrth}
Klimm, M. and Warode, P. (2021).
\newblock Parametric computation of minimum-cost flows with piecewise quadratic
  costs.
\newblock {\em Math. Oper. Res.}, forthcoming.

\bibitem[Koutsoupias and Papadimitriou, 1999]{KouPap:STACS1999}
Koutsoupias, E. and Papadimitriou, C. (1999).
\newblock Worst-case equilibria.
\newblock In {\em S{TACS} 99 ({T}rier)}, volume 1563 of {\em Lecture Notes in
  Comput. Sci.}, pages 404--413. Springer, Berlin.

\bibitem[Milchtaich, 2006]{Mil:GEB2006}
Milchtaich, I. (2006).
\newblock Network topology and the efficiency of equilibrium.
\newblock {\em Games Econom. Behav.}, 57(2):321--346.

\bibitem[Monnot et~al., 2017]{MonBenPil:WINE2018}
Monnot, B., Benita, F., and Piliouras, G. (2017).
\newblock Routing games in the wild: efficiency, equilibration and regret.
\newblock In R.~Devanur, N. and Lu, P., editors, {\em Web and Internet
  Economics}, pages 340--353, Cham. Springer International Publishing.

\bibitem[O'Hare et~al., 2016]{OHaConWat:TRB2016}
O'Hare, S.~J., Connors, R.~D., and Watling, D.~P. (2016).
\newblock Mechanisms that govern how the price of anarchy varies with travel
  demand.
\newblock {\em Transportation Res. Part B}, 84:55--80.

\bibitem[Papadimitriou, 2001]{Pap:ACMSTC2001}
Papadimitriou, C. (2001).
\newblock Algorithms, games, and the {Internet}.
\newblock In {\em Proceedings of the {T}hirty-{T}hird {A}nnual {ACM}
  {S}ymposium on {T}heory of {C}omputing}, pages 749--753, New York. ACM.

\bibitem[Patriksson, 2004]{Pat:TS2004}
Patriksson, M. (2004).
\newblock Sensitivity analysis of traffic equilibria.
\newblock {\em Transportation Sci.}, 38(3):258--281.

\bibitem[Pigou, 1920]{Pig:Macmillan1920}
Pigou, A.~C. (1920).
\newblock {\em The Economics of Welfare}.
\newblock Macmillan and Co., London, 1st edition.

\bibitem[Rockafellar, 1997]{Roc:PUP1997}
Rockafellar, R.~T. (1997).
\newblock {\em Convex Analysis}.
\newblock Princeton Landmarks in Mathematics. Princeton University Press,
  Princeton, NJ.
\newblock Reprint of the 1970 original, Princeton Paperbacks.

\bibitem[Roughgarden, 2003]{Rou:JCSS2003}
Roughgarden, T. (2003).
\newblock The price of anarchy is independent of the network topology.
\newblock {\em J. Comput. System Sci.}, 67(2):341--364.

\bibitem[Roughgarden and Tardos, 2002]{RouTar:JACM2002}
Roughgarden, T. and Tardos, E. (2002).
\newblock How bad is selfish routing?
\newblock {\em J. ACM}, 49(2):236--259.

\bibitem[Roughgarden and Tardos, 2004]{RouTar:GEB2004}
Roughgarden, T. and Tardos, E. (2004).
\newblock Bounding the inefficiency of equilibria in nonatomic congestion
  games.
\newblock {\em Games Econom. Behav.}, 47(2):389--403.

\bibitem[Shapiro, 1988]{Sha:SIAMJCO1988}
Shapiro, A. (1988).
\newblock Sensitivity analysis of nonlinear programs and differentiability
  properties of metric projections.
\newblock {\em SIAM J. Control Optim.}, 26(3):628--645.

\bibitem[Takalloo and Kwon, 2020]{TakKwo:OL2020}
Takalloo, M. and Kwon, C. (2020).
\newblock Sensitivity of {W}ardrop equilibria: revisited.
\newblock {\em Optim. Lett.}, 14(3):781--796.

\bibitem[Tomlin, 1966]{Tom:OR1966}
Tomlin, J.~A. (1966).
\newblock Minimum-cost multicommodity network flows.
\newblock {\em Oper. Res.}, 14(1):45--51.

\bibitem[Wardrop, 1952]{War:PICE1952}
Wardrop, J.~G. (1952).
\newblock Some theoretical aspects of road traffic research.
\newblock In {\em Proceedings of the Institute of Civil Engineers, Part II},
  volume~1, pages 325--378.

\bibitem[Wu and M\"ohring, 2020]{WuMoh:arXiv2020}
Wu, Z. and M\"ohring, R. (2020).
\newblock A sensitivity analysis for the price of anarchy in non-atomic
  congestion games.
\newblock Technical report, arXiv 2007.13979.

\bibitem[Wu et~al., 2021]{WuMohCheXu:OR2021}
Wu, Z., M{\"o}hring, R.~H., Chen, Y., and Xu, D. (2021).
\newblock Selfishness need not be bad.
\newblock {\em Operations Research}, 69(2):410--435.

\bibitem[Youn et~al., 2008]{YouGasJeo:PRL2008}
Youn, H., Gastner, M.~T., and Jeong, H. (2008).
\newblock Price of anarchy in transportation networks: Efficiency and
  optimality control.
\newblock {\em Physical Review Letters}, 101(12):128701.

\bibitem[Youn et~al., 2009]{YouGasJeo:PRL2009}
Youn, H., Gastner, M.~T., and Jeong, H. (2009).
\newblock Erratum: Price of anarchy in transportation networks: Efficiency and
  optimality control [phys. rev. lett. \textbf{101}, 128701 (2008)].
\newblock {\em Phys. Rev. Lett.}, 102:049905.

\end{thebibliography}

\end{document}